\pdfoutput=1
\RequirePackage{ifpdf}
\ifpdf 
\documentclass[pdftex]{sigma}
\else
\documentclass{sigma}
\fi

\newcommand{\be}{\begin{eqnarray}}
\newcommand{\ee}{\end{eqnarray}}
\newcommand{\bez}{\begin{eqnarray*}}
\newcommand{\eez}{\end{eqnarray*}}
\newcommand{\cA}{\mathcal{A}}
\newcommand{\tT}{\tilde{T}}
\newcommand{\cS}{S}
\newcommand{\cP}{\mathcal{P}}

\newcommand{\cU}{\mathcal{U}}
\newcommand{\cV}{\mathcal{V}}
\newcommand{\bsy}{\boldsymbol}

\newcommand{\bbN}{\mathbb{N}}
\newcommand{\bbR}{\mathbb{R}}
\newcommand{\bbZ}{\mathbb{Z}}
\renewcommand\vec[1]{\overrightarrow{#1}}
\newcommand\cev[1]{\overleftarrow{#1}}

\newtheorem{Theorem}{Theorem}[section]
\newtheorem{Proposition}[Theorem]{Proposition}
\newtheorem{Corollary}[Theorem]{Corollary}
\newtheorem{Conjecture}[Theorem]{Conjecture}
\theoremstyle{definition}
\newtheorem{Remark}[Theorem]{Remark}
\newtheorem{Example}[Theorem]{Example}

\numberwithin{equation}{section}

\begin{document}

\allowdisplaybreaks

\newcommand{\arXivNumber}{2305.17974}

\renewcommand{\PaperNumber}{051}

\FirstPageHeading

\ShortArticleName{On the Structure of Set-Theoretic Polygon Equations}

\ArticleName{On the Structure of Set-Theoretic Polygon Equations}

\Author{Folkert M\"ULLER-HOISSEN}

\AuthorNameForHeading{Folkert M\"uller-Hoissen}

\Address{Institut f\"ur Theoretische Physik, Friedrich-Hund-Platz 1, 37077 G\"ottingen, Germany}
\Email{\href{mailto:folkert.mueller-hoissen@phys.uni-goettingen.de}{folkert.mueller-hoissen@phys.uni-goettingen.de}}
\URLaddress{\url{http://mueller-hoissen.math-phys.org}}

\ArticleDates{Received December 29, 2023, in final form May 29, 2024; Published online June 11, 2024}

\Abstract{Polygon equations generalize the prominent \emph{pentagon equation} in very much the same way as simplex equations generalize the famous \emph{Yang--Baxter equation}. In particular, they appeared as ``cocycle equations'' in Street's category theory associated with oriented simplices. Whereas the $(N-1)$-simplex equation can be regarded as a realization of the higher Bruhat order $B(N,N-2)$, the $N$-gon equation is a realization of the higher Tamari order $T(N,N-2)$. The latter and its dual $\tilde T(N,N-2)$, associated with which is the dual $N$-gon equation, have been shown to arise as suborders of $B(N,N-2)$ via a ``three-color decomposition''. There are two different reductions of $T(N,N-2)$ and $\tilde T(N,N-2)$, to~${T(N-1,N-3)}$, respectively $\tilde T(N-1,N-3)$. In this work, we explore the corresponding reductions of (dual) polygon equations, which lead to relations between solutions of neighboring (dual) polygon equations. We also elaborate (dual) polygon equations in this respect explicitly up to the octagon equation.}


\Classification{06A06; 06A07; 15A69; 16T05; 16T25; 17A01; 18D10}

\section{Introduction}
An infinite family of ``polygon equations'' has been introduced in \cite{DMH15}, based on the
combinatorial structure of higher Tamari (partial) orders. They appeared earlier as cocycle equations
(``higher cocycloids'') in a category-theoretical framework based on oriented simplices (``orientals'') \cite{Stre87}.
The 4-gon equation is the condition of associativity for a binary operation. The dual~4-gon equation expresses coassociativity of
a comultiplication. The 5-gon or pentagon equation is the most prominent member of this family and
plays a profound role in mathematics and mathematical physics.

Of particular relevance is the pentagon equation in the context of bi- and Hopf algebras. The crucial and fairly simple observation here is the following. Given a unital associative algebra~$\cA$
with identity element $1_\cA$, it carries a trivial (left) comultiplication\footnote{Of course, there is a corresponding
result for the ``right'' trivial comultiplication $\Delta_r\colon \cA \rightarrow \cA \otimes \cA$,
$a \mapsto 1_\cA \otimes a$. }
${\Delta_\ell\colon \cA \rightarrow \cA \otimes \cA}$,~${a \mapsto a \otimes 1_\cA}$.
Non-trivial comultiplications are then obtained via $\Delta := W \Delta_\ell W^{-1}$ if $W \in \cA \otimes \cA$
is an invertible solution of the pentagon equation\footnote{A somewhat more general result
can be found in \cite{Davy01} and involves a ``mixed'' or ``entwining'' pentagon equation. Corresponding
versions of polygon equations, involving several different maps, have been treated in full
generality in \cite{DMH15}. Also see~\cite{Kass23}, for example.}
$W_{\bsy{12}} W_{\bsy{13}} W_{\bsy{23}} = W_{\bsy{23}} W_{\bsy{12}}$
in the threefold tensor product of~$\cA$, where the indices determine at which two components of it $W$ acts.
In fact, each comultiplication on a finite-dimensional algebra can be expressed in this way~\cite{Mili04}.
But also in the rigorous framework of infinite-dimensional $C^\ast$-algebras such an expression for comultiplication
plays a~profound role in surpassing technical problems in the duality theory of locally compact groups and for establishing
a~rigorous setting for quantum groups. Here $W$ appears as a unitary linear
operator under the name ``multiplicative unitary'' or (generalized) Kac--Takesaki operator
(see, in particular, \cite{Baaj+Skan93,Kust+Vaes00,Skan91,Take71,Take72,Timm08,Woro96}).

The pentagon equation, sometimes called ``fusion equation'', shows up in conformal field theory as a
condition for the fusion matrix, which arises from operator product expansion of chiral vertex operators,
see \cite{Moor+Seib89} for example.

The pentagon equation plays a crucial role as an associativity constraint for the basic functor of a monoidal
category, notably because of Mac Lane's coherence theorem \cite{MacL63} (also see~\cite{Kell64}).
This in turn has roots in Tamari's
work on relaxing associativity and Stasheff's work on homotopy associativity \cite{Stash12,Stash63}.

One way to obtain topological invariants of manifolds is via triangulations, i.e., decompositions into simplices.
To each simplex one assigns a mathematical object and builds from all of them an object (e.g., a ``state sum'')
assigned to a triangulation.
This has to be done in such a way that the latter object is invariant under manipulations of triangulations
called Pachner moves or bistellar flips \cite{Pach91} (also see, e.g., \cite{Lick99}).
In three dimensions, the simplices are tetrahedra. A Pachner $(2,3)$-move splits two adjacent tetrahedra into
three and this requires that the objects associated with the tetrahedra, constituting a triangulation,
satisfy a pentagon relation. Because of the pentagonal Biedenharn--Elliot identity \cite{Bied+Louc81},
a Wigner $6$-$j$ symbol, or a~generalization of it, is thus a choice for the object assigned
to a tetrahedron \cite{Ponz+Regg68,Tura94,Tura+Viro92}.\footnote{Essentially, (generalized) $6$-$j$ symbols arise
as follows. If there is a direct sum decomposition
$V_i \otimes V_j = \bigoplus_\ell H^\ell_{ij} \otimes V_\ell$ of vector spaces,
the associator isomorphism $(V_i \otimes V_j) \otimes V_k \rightarrow V_i \otimes (V_j \otimes V_k)$
of a monoidal category induces maps
\smash{$\left\{ \begin{smallmatrix} i & j & \ell \\ k & m & n \end{smallmatrix} \right\} \colon H^\ell_{ij} \otimes H^m_{\ell k}
 \rightarrow H^m_{i n} \otimes H^n_{j k} $}.
}
In \cite{Barr+Cran97}, it was shown that the Wheeler--DeWitt equation for three-dimensional general relativity
reduces to the pentagon relation.
There is a comprehensive literature around the idea of representing Pachner moves of a triangulation
in three dimensions by a solution of the pentagon equation (or a pentagon relation), see in particular \cite{Bara+Frei15, Cran+Fren94,Kash94,Kore02,MST22,Suzu18}.
In four dimensions, the dual hexagon equation plays a corresponding role \cite{Kash15} (also see Remark~\ref{rem:Kashaev_Pachner4d} below), in five dimensions it is the dual heptagon equation
\cite{Kore21,Kore22a,Kore22b}.

In the latter context, it is also worth to mention that the higher Tamari orders, introduced in~\cite{DMH15} and underlying polygon equations, are equivalent \cite{Will23} to higher Stasheff--Tamari orders~\mbox{\cite{Edel+Rein96,Kapr+Voev91}} on the set of triangulations of a cyclic polytope.

A pentagonal relation is also satisfied by the Rogers \cite{Roge07} and quantum dilogarithms\footnote{Also the
quantum dilogarithm can be understood as a $6$-$j$ symbol \cite{Kash94}.} \cite{Fadd11,Fadd+Kash94}.
Although it does not have the structure of the pentagon equation as considered in this work, there
are certain relations, see, e.g., Example~\ref{ex:Rogers_dilog} below.\footnote{Also Drinfeld associators of quasi-Hopf
algebras (see, e.g., \cite{Etin+Schi98}) are subject to a certain pentagon equation, but they are
elements of a triple tensor product and therefore the corresponding pentagon equation is different from the standard one.}

Comparatively little is known so far about higher polygon equations. An expression for the dual hexagon equation appeared as a 4-cocycle condition in \cite{Stre87,Stre98}.

Polygon equations play a role in the theory of completely solvable models of statistical mechanics.
If solutions of the $N$-gon and the dual $N$-gon equation satisfy a certain compatibility condition, a special solution of the $(N-1)$-simplex equation is obtained, see \cite[Section~5]{DMH15}. This includes and generalizes a relation between solutions of the pentagon equation, its dual, and the~4-simplex (Bazhanov--Stroganov) equation \cite{Kash+Serg98}. There is also a relation between the pentagon equation, its dual, and the 3-simplex (tetrahedron or Zamolodchikov) equation \cite{Kash+Serg98,Kass23, Mail94,Serg23}, for which a corresponding generalization to higher polygon and simplex equations is yet unknown.

A class of solutions of odd polygon equations with maps acting on direct sums has recently been
obtained in \cite{Dima+Kore21}.

In the present work, we explore polygon equations in the set-theoretic setting, i.e., we do not assume any additional structure beyond sets and maps between (Cartesian products of) them. Using the combinatorics of higher Tamari orders, underlying polygon equations, we derive relations between the sets of solutions of neighboring (dual) polygon equations. These relations are proved for the whole infinite family of (dual) polygon equations. Additional results are provided for (dual) polygon equations up to the octagon (8-gon) equation.

Section~\ref{sec:B&T} briefly recalls the definition of higher Bruhat and higher Tamari orders. The reader is
referred to \cite{DMH15} for a detailed treatment. Section~\ref{sec:S&P} presents the definition of simplex and polygon equations in a slightly different way from \cite{DMH15}.
Section~\ref{sec:reductions} deals with reductions of polygon equations and contains the most general results
of this work.
Sections~\ref{sec:ExPEs} and \ref{sec:dualPol} present the first few examples of polygon equations and their
duals, respectively. They extend results obtained in \cite{DMH15}.\looseness=-1

Section~\ref{sec:relations_polygon_eqs} deals with the special, but most important case of (dual) polygon equations
for a~single map, acting between Cartesian products of a set $\cU$. Partly from the general
results in Section~\ref{sec:reductions}, but to some extent also independently and from results in Sections~\ref{sec:ExPEs} and \ref{sec:dualPol}, we derive relations between solutions of a (dual) $N$-gon and (dual) $(N+1)$-gon equation, $N>2$.

Section~\ref{sec:conclusions} contains some concluding remarks.

\section{Basics of higher Bruhat and higher Tamari orders}
\label{sec:B&T}
For a non-empty finite subset $M$ of $\bbN$, and $n \in \bbN$, $1 \leq n \leq |M|$
(where $|M|$ is the cardinality of $M$), let ${M \choose n}$ denote the set of
$n$-element subsets of $M$.
The \emph{packet} $P(M)$ of $M$ is the set of~${(|M|-1)}$-element subsets of $M$. We write
\smash{$\vec{P}(M)$} for $P(M)$ in lexicographic order, and
\smash{$\cev{P}(M)$} for $P(M)$ in reverse lexicographic order.

Let $N \in \bbN$, $N>1$, and $[N]=\{1,\ldots, N\}$. A \emph{linear order} (permutation) $\rho$
of ${[N] \choose n}$, ${n\in \bbN}$, ${n<N}$, is called \emph{admissible} if, for each
\smash{$K \in {[N] \choose n+1}$}, the packet $P(K)$ is contained in $\rho$ in lexicographic or in reverse lexicographic order.
Let $A(N,n)$ denote the set of admissible linear orders of ${[N] \choose n}$.
An equivalence relation is defined on $A(N,n)$ by $\rho \sim \rho'$ if and only if $\rho$ and~$\rho'$ only differ by exchange
of two neighboring elements, not both contained in some packet. Then~${A(N,n)/{\sim}}$, supplied with the
partial order given via inversions of lexicographically ordered packets, \smash{$\vec{P}(K) \mapsto \cev{P}(K)$},
is the \emph{higher Bruhat order} $B(N,n)$.

Next we consider the splitting of a packet, $P(K) = P_{\rm o}(K) \cup P_{\rm e}(K)$, where $P_{\rm o}(K)$ ($P_{\rm e}(K)$) is the half-packet
consisting of elements with odd (even) position in the lexicographically ordered~$P(K)$.

We say an element $J \in P_{\rm o}(K)$ is \emph{blue} in \smash{$\vec{P}(K)$} and \emph{red} in \smash{$\cev{P}(K)$},
an element $J \in P_{\rm e}(K)$ is \emph{red} in \smash{$\vec{P}(K)$} and \emph{blue} in \smash{$\cev{P}(K)$}.
$J \in P(K)$ is \emph{blue} (\emph{red}) in $\rho \in A(N,n)$ if $J$ is \emph{blue} (\emph{red}) with respect to
all $K$ for which $J \in P(K)$ and either \smash{$\vec{P}(K)$} or \smash{$\cev{P}(K)$} is a subsequence of $\rho$.

It can happen that $J$ is \emph{blue} in $\rho \in A(N,n)$ with respect to some $K$ and \emph{red} with respect to another $K'$. In such a case we color it \emph{green}. By $\rho^{\rm (b)}$, $\rho^{\rm (r)}$, $\rho^{\rm (g)}$
we denote the blue, red, respectively green subsequence of~$\rho$.

It has been shown in \cite{DMH15} that there are projections $B(N,n) \rightarrow B^{\rm (c)}(N,n)$, $[\rho] \mapsto \big[\rho^{\rm (c)}\big]$,
${\rm c} \in \{{\rm b},{\rm r},{\rm g}\}$, such that $B^{\rm (c)}(N,n)$ inherits a partial order from $B(N,n)$.
$T(N,n) := B^{\rm (b)}(N,n)$ are the \emph{higher Tamari orders} and $\tT(N,n) := B^{\rm (r)}(N,n)$ are called
\emph{dual higher Tamari orders}. The inversion operation in case of $T(N,n)$ is
\smash{$\vec{P}_{\rm o}(K) \mapsto \cev{P}_{\rm e}(K)$}, \smash{$K \in {[N] \choose n+1}$}. In case of $\tT(N,n)$, it is
\smash{$\vec{P}_{\rm e}(K) \mapsto \cev{P}_{\rm o}(K)$}.

\begin{Remark}
Associating with $K \in {[N] \choose n+1}$ an $n$-simplex, the packet $P(K)$ corresponds to the set of its faces, which
are $(n-1)$-simplices. The \emph{pasting scheme} given by the above inversion operation then supplies the faces
with an orientation.
This results in the \emph{orientals} (oriented simplices) introduced by Street in 1987 \cite{Stre87}.
It had been conjectured in \cite{DMH12KPBT,DMH15} and proved in~\cite{Will23} that the higher Tamari orders
are equivalent to the \emph{higher Stasheff–Tamari orders} in~\mbox{\cite{Edel+Rein96,Kapr+Voev91}}.
All these works thus deal
with essentially the same structure. In \cite{DMH11KPT,DMH12KPBT} it has been realized in terms of rooted tree-shaped
solutions of the \emph{Kadomtsev--Petviashvili $($KP$)$ hierarchy}.
\end{Remark}

\section{A brief account of simplex and polygon equations}
\label{sec:S&P}
Let $N>2$. With \smash{$J \in {[N] \choose N-2}$} we associate a set
$\cU_J$. For $\rho \in A(N,N-2)$, let $\cU_\rho$ be the correspondingly ordered
Cartesian product of the $\cU_J$, $J \in \rho$.
With \smash{$K \in {[N] \choose N-1} = P([N])$} we associate a map
\[
 R_K \colon\ \cU_{\vec{P}(K)} \longrightarrow \cU_{\cev{P}(K)}.
\]
The \emph{$(N-1)$-simplex equation}
\begin{gather}
 R_{\vec{P}([N])} = R_{\cev{P}([N])} \label{(N-1)-simplex_eq}
\end{gather}
may then be regarded as a realization of $B(N,N-2)$.
The expressions on both sides are compositions of maps $R_K$, $K \in P([N])$,
applied on the left-hand side in lexicographic, on the right-hand side in reverse lexicographic order.
Writing \smash{$\vec{P}([N]) = (K_1,\ldots,K_N)$},
we have \smash{$R_{\vec{P}([N])} = R_{K_N} \cdots R_{K_1}$}.
Hence, as a composition, \smash{$R_{\vec{P}([N])}$} is actually in reverse lexicographic order,
but the maps are applied in lexicographic order.
We have to add the following rules in order for \eqref{(N-1)-simplex_eq} to make sense.
\begin{enumerate}\itemsep=0pt
\item[(1)] Both sides of \eqref{(N-1)-simplex_eq} act on $\cU_\alpha$ and map to
$\cU_\omega$, where $\alpha$ ($\omega$) is ${[N] \choose N-2}$ in lexicographic (reverse lexicographic) order.
\item[(2)] Each of the maps $R_K$ acts at consecutive positions in the respective multiple Cartesian product of spaces.
\item[(3)] If $J, J' \in {[N] \choose N-2}$ are such that they do not both belong to $P(K)$ for any
$K \in {[N] \choose N-1}$, then
\[
 \cdots \times \cU_J \times \cU_{J'} \times \cdots \; \boldsymbol{\sim} \;
 \cdots \times \cU_{J'} \times \cU_J \times \cdots
\]
imposes an equivalence relation on Cartesian products.
\end{enumerate}

Starting with $\cU_\alpha$, it may be necessary to use the third rule to arrange
that $R_{K_1}$, respectively~$R_{K_N}$, can be applied, which means that the sets associated with
elements of $P(K_1)$, respectively~$P(K_N)$, have to be in lexicographic order and at neighboring
positions in the respective multiple Cartesian product of sets.
After an application of some $R_K$, it may again be necessary to use the third rule to arrange
a further application of a map $R_{K'}$, or to achieve the final reverse lexicographic order
$\cU_\omega$.
That this works is a consequence of the underlying structure of higher Bruhat orders \cite{DMH15,Manin+Schecht86a,Manin+Shekhtman86b}.

 We have to stress that \eqref{(N-1)-simplex_eq} is \emph{not} the form in which simplex equations usually
appear in the literature, see \cite{DMH15} for the relation and references.

With each $K \in {[N] \choose N-1}$, now we associate a map
\[
 T_K\colon\ \cU_{\vec{P}_{\rm o}(K)} \longrightarrow \cU_{\cev{P}_{\rm e}(K)}.
\]
Writing $K = (k_1,\ldots,k_{N-1})$, with $k_i < k_{i+1}$, $i=1,\ldots,N-2$, we have
\begin{align*}
 \cU_{\vec{P_{\rm o}}(K)}
 &= \cU_{ K \setminus \{k_{N-1}\} } \times \cU_{ K \setminus \{k_{N-3}\} } \times \cdots \times
 \cU_{ K \setminus \{k_{1 + (N \, \mathrm{mod}\, 2)}\} }, \\
 \cU_{\cev{P_{\rm e}}(K)}
 &= \cU_{ K \setminus \{k_{2 - (N \, \mathrm{mod}\, 2)}\} } \times \cdots \times
 \cU_{ K \setminus \{k_{N-4}\} } \times \cU_{ K \setminus \{k_{N-2}\} }.
\end{align*}

The \emph{$N$-gon equation}
\begin{gather}
 T_{\vec{P}_{\rm o}([N])} = T_{\cev{P}_{\rm e}([N])} \label{N-gon_eq_short}
\end{gather}
may be regarded as a realization of $T(N,N-2)$.
It is well defined if we require the following rules.
\begin{enumerate}\itemsep=0pt
\item[(1)] Let $\alpha$ ($\omega$) be again \smash{${[N] \choose N-2}$} in lexicographic (reverse lexicographic) order,
and let $\alpha^{\rm (b)}$ and~\smash{$\omega^{\rm (b)}$} be the corresponding blue parts.
Both sides of \eqref{N-gon_eq_short} act on $\cU_{\alpha^{\rm (b)}}$ and map to
$\cU_{\omega^{\rm (b)}}$.
\item[(2)] Each of the maps $T_K$ acts at consecutive positions in the respective multiple Cartesian product of
sets.
\item[(3)] If $J, J' \in {[N] \choose N-2}$ are such that they do not both belong to $P(K)$ for any
$K \in {[N] \choose N-1}$, then
\[
 \cdots \times \cU_J \times \cU_{J'} \times \cdots \; \boldsymbol{\sim} \;
 \cdots \times \cU_{J'} \times \cU_J \times \cdots.
\]
\end{enumerate}

As in the case of simplex equations, to apply or work out a polygon equation, we have
to check at each step whether a map $T_K$ can be applied directly or whether we first have to use
the third rule above to achieve a reordering of the respective multiple Cartesian product.
In any case, we have to keep track of the numbering of the sets, even if they are identical as sets.

It is, therefore, convenient to realize the above equivalence relation by introducing explicitly
transposition maps (sometimes called flip or switch maps) in the equations, at the price of ending up with a form
of the equation that looks more complicated and apparently lost its universal structure, but it is
often better suited for applications.
Instead of keeping track of the numbering of sets, we then have to keep track of the first
position on which a map acts in a multiple Cartesian product. This has been done in \cite{DMH15}.
For several polygon equations, we will recall the resulting form in Section~\ref{sec:ExPEs}.
In this form, we can best deal with the case of prime interest, where all the sets $\cU_J$ are the
same and there is only a single map $T$.

If $N$ is odd, \eqref{N-gon_eq_short} can be written as
\begin{gather}
 T_{\hat{1}} T_{\hat{3}} \cdots T_{\widehat{N-2}} T_{\widehat{N}} = T_{\widehat{N-1}} T_{\widehat{N-3}} \cdots T_{\hat{2}},
 \label{N-gon_odd}
\end{gather}
where $\hat{k} := [N] \setminus \{k\}$ (complementary index notation).

If $N$ is even, \eqref{N-gon_eq_short} can be correspondingly expressed as
\begin{gather}
 T_{\hat{2}} T_{\hat{4}} \cdots T_{\widehat{N-2}} T_{\widehat{N}} = T_{\widehat{N-1}} T_{\widehat{N-3}} \cdots T_{\hat{1}}.
 \label{N-gon_even}
\end{gather}

With each $K \in {[N] \choose N-1}$, $N>2$, we also associate a map
\[
 \tT_K \colon\ \cU_{\vec{P}_{\rm e}(K)} \longrightarrow \cU_{\cev{P}_{\rm o}(K)}.
\]
The \emph{dual $N$-gon equation}
\begin{gather}
 \tT_{\vec{P}_{\rm e}([N])} = \tT_{\cev{P}_{\rm o}([N])} \label{N-gon_eq}
\end{gather}
may be regarded as a realization of the dual Tamari order $\tT(N,N-2)$. Both sides act on $\alpha^{\rm (r)}$,
which is equal to $\omega^{\rm (b)}$ totally reversed, and map to $\omega^{\rm (r)}$, which is equal to $\alpha^{\rm (b)}$
totally reversed.

For odd $N$, \eqref{N-gon_eq} is
\[ 
 \tT_{\hat{2}} \tT_{\hat{4}} \cdots \tT_{\widehat{N-1}}
 = \tT_{\widehat{N}} \tT_{\widehat{N-2}} \cdots \tT_{\hat{3}} \tT_{\hat{1}},
\]
which is \eqref{N-gon_odd} reversed.
For even $N$, we have
\begin{gather}
 \tT_{\hat{1}} \tT_{\hat{3}} \cdots \tT_{\widehat{N-1}}
 = \tT_{\widehat{N}} \tT_{\widehat{N-2}} \cdots \tT_{\hat{2}}, \label{dual_N-gon_even}
\end{gather}
which is \eqref{N-gon_even} reversed.

Simplex equations, and also (dual) polygon equations, are interrelated by a kind of integrability feature, which
crucially distinguishes them from similar equations. We refer to \cite{DMH15} for the general structure, but
in Section~\ref{sec:ExPEs} we elaborate this feature for some examples of polygon equations.

\begin{Remark}
The dual $(N+2)$-gon equation is the $N$-cocycle condition in \cite{Stre87,Stre98}. In \cite{Stre87}, the reader finds
an explanation in which sense these equations can be regarded as ``cocycles''.
\end{Remark}

\section{Reductions of polygon equations}
\label{sec:reductions}

For any fixed $k \in [N+1]$, there is a projection $\mathrm{pr}_k\colon A(N+1,n+1) \rightarrow A(N,n)$, obtained by restricting
$\rho \in A(N+1,n+1)$ to the subsequence consisting only of elements $K \in {[N+1] \choose n+1}$ with~${k \in K}$.
The set of all these subsequences
is in bijection with $A(N,n)$, simply by deleting $k$ in each~$K$ and an obvious renumbering. Moreover, the projection
is compatible with the equivalence relation~$\sim$ and induces
a projection $B(N+1,n+1) \rightarrow B(N,n)$. See~\cite[Remark~2.5]{DMH15}.

But only for $k \in \{1,N+1\}$, the projection $\mathrm{pr}_k $ is compatible with the 3-color decomposition,
see \cite[Remark~2.17]{DMH15}.
If $k=1$, this yields projections $T(N+1,n+1) \rightarrow T(N,n)$ and
${\tT(N+1,n+1) \rightarrow \tT(N,n)}$. For $k=N+1$, we have
projections $T(N+1,n+1) \rightarrow \tT(N,n)$ and $\tT(N+1,n+1) \rightarrow T(N,n)$.

In particular, there are thus projections $T(N+1,N-1) \rightarrow T(N,N-2)$ and $T(N+1,N-1) \rightarrow \tT(N,N-2)$, which then induce reductions of the $(N+1)$-gon to the $N$-gon equation, respectively the dual $N$-gon equation.
In the same way, the projections $\tT(N+1,N-1) \rightarrow \tT(N,N-2)$ and~${\tT(N+1,N-1) \rightarrow T(N,N-2)}$
induce reductions of the dual $(N+1)$-gon to the dual $N$-gon equation, respectively the~$N$-gon equation.
Since we realize Tamari orders by sets and maps between them, we have to make sure, however, that this procedure indeed leads to a~realization of the (dual) $N$-gon equation. This will be made precise in the following subsections. Throughout we assume $N>2$.

By a \emph{degenerate} map (e.g., a solution of a polygon equation), we mean a map whose values do not
depend on (at least) one of its arguments.

\subsection[Reductions induced by pr\_1]{Reductions induced by $\boldsymbol{\mathrm{pr}_1}$}
The reduction of the $(N+1)$-gon equation induced by $\mathrm{pr}_1$ is essentially
obtained by dropping the map $T_K$ with $1 \notin K$. But we also have to arrange that the remaining maps $T_K$ with $ 1 \in K$ are reduced to maps between (products of) spaces $U_J$ with $1 \in J \in P(K)$.

\begin{Theorem}
\label{thm:polygon->oddpolygon_red1}
Let $N \in \bbN$ be odd, $N>2$.
\begin{itemize}\itemsep=0pt
\item[$(1)$] Let $T_K$, $K \in {[N+1] \choose N}$, be maps solving the $(N+1)$-gon equation.
For $K$ with $1 \in K$, let~$T_K$ not depend on the last component of its domain. Let $T_{K'}$, $K' = K \setminus \{1\}$, be obtained from $T_K$ by excluding $\cU_{K'}$ from its domain.
Then \smash{$\big\{T_{K'} | K' \in {\{2,\ldots, N+1 \} \choose N-1 } \big\}$} solve the $N$-gon equation.
\item[$(2)$] Each solution of the $N$-gon equation can be extended to a degenerate solution of the $(N+1)$-gon equation.
\end{itemize}
\end{Theorem}
\begin{proof}
(1) If $N$ is odd, then $ K' = K \setminus \{1\}$ is the last element of $\vec{P}_{\rm o}(K)$. If
\[
 T_K \colon\ \cU_{\vec{P}'_{\rm o}(K)} \times \cU_{K \setminus \{1\}} \rightarrow \cU_{\cev{P}_{\rm e}(K)}
\]
does not depend on the last component of its domain, it induces a map
\[
 T_{K'} \colon\ \cU_{\vec{P}_{\rm o}(K')} \rightarrow \cU_{\cev{P}_{\rm e}(K')},
\]
with \smash{$\cU_J := \cU_{\{1\} \cup J}$} for $J \in P(K')$.
Since \smash{$T_{\hat{1}} = T_{\{2,3,\ldots,N+1\}}\colon \cU_{\vec{P}_{\rm o}(\{2,3,\ldots, N+1\})} \rightarrow \cU_{\cev{P}_{\rm e}(\{2,3,\ldots, N+1\})}$} maps to spaces that are, as a consequence of our assumption, disregarded by all other maps,
it can be dropped from the $(N+1)$-gon equation
\[
 T_{\hat{2}} T_{\hat{4}} \cdots T_{\widehat{N-1}} T_{\widehat{N+1}} = T_{\widehat{N}} T_{\widehat{N-2}} \cdots T_{\hat{1}},
\]
which thus reduces to
\[
 T_{\hat{2}} T_{\hat{4}} \cdots T_{\widehat{N-1}} T_{\widehat{N+1}} = T_{\widehat{N}} T_{\widehat{N-2}} \cdots T_{\hat{3}},
\]
where we can now regard the indices as complementary in $[N+1] \setminus \{1\}$.
By a shift in the numbering, this is turned into the standard form of the odd $N$-gon equation,
\[
 T_{\hat{1}} T_{\hat{3}} \cdots T_{\widehat{N-2}} T_{\widehat{N}} = T_{\widehat{N-1}} T_{\widehat{N-3}} \cdots T_{\hat{2}},
\]
where the indices are complementary in $[N]$.

(2) Applying a shift in the numbering of maps, constituting a solution of the $N$-gon equation, it is given by
\smash{$\big\{T_{K'} | K' \in {\{2,\ldots, N+1 \} \choose N-1 } \big\}$}. Associating $K=\{1\} \cup K'$ with $K'$, we trivially extend
the map~$T_{K'}$ to a map $T_K$ of the form given above, by introducing a set
$\cU_{K \setminus \{1\}}$. Choosing furthermore a~map \smash{$T_{\{2,3,\ldots, N+1\}} \colon \cU_{\vec{P}_{\rm o}(\{2,3,\ldots, N+1\})} \rightarrow \cU_{\cev{P}_{\rm e}(\{2,3,\ldots, N+1\})}$}, with sets $\cU_J$, $J \in P(\{2,3,\ldots, N+1\})$, we obtain a solution of the $(N+1)$-gon equation. This essentially reverses the steps taken in the proof of (1).
\end{proof}

\begin{Theorem}
\label{thm:polygon->evenpolygon_red1}
Let $N \in \bbN$ be even, $N>2$. Let $T_K$, $K \in {[N+1] \choose N}$, be maps solving the $(N+1)$-gon equation.
For $K$ with $1 \in K$ let $T_{K'}$, $K' = K \setminus \{1\}$, be obtained from $T_K$ by deleting the first
component of its codomain. Then \smash{$\bigl\{T_{K'} | K' \in {\{2,\ldots, N+1 \} \choose N-1 } \bigr\}$}
solve the $N$-gon equation.
\end{Theorem}
\begin{proof}
Let $N$ be even and $K$ such that $1 \in K$. The packet $P(K)$ has $N$ elements and its last member
in the lexicographic order is $K \setminus \{1\}$. Hence, $K \setminus \{1\}$ is the first element of \smash{$\cev{P}_{\rm e}(K)$}.
Then
\[
 T_K \colon\ \cU_{\vec{P}_{\rm o}(K)} \rightarrow \cU_{K \setminus \{1\}} \times \cU_{\cev{P}'_{\rm e}(K)},
\]
where \smash{$\cev{P}'_{\rm e}(K)$} is \smash{$\cev{P}_{\rm e}(K)$} without the element $K' := K \setminus \{1\}$.
By disregarding the first component of its codomain, each $T_K$, $1 \in K$, induces a map
\[
 T_{K'} \colon\ \cU_{\vec{P}_{\rm o}(K')} \rightarrow \cU_{\cev{P}_{\rm e}(K')}, \qquad K' \in { \{2,\ldots N+1 \} \choose N-1 },
\]
where, for $J \in P(K')$, we set again $\cU_J := \cU_{\{1\} \cup J}$.
Since the domain of the remaining map $T_{\hat{1}}$ only involves spaces that are excluded from the
range of all other maps $T_K$, it splits off from the $(N+1)$-gon equation
\[
 T_{\hat{1}} T_{\hat{3}} \cdots T_{\widehat{N-1}} T_{\widehat{N+1}} = T_{\widehat{N}} T_{\widehat{N-2}} \cdots T_{\hat{2}},
\]
which thus reduces to
\[
 T_{\hat{3}} \cdots T_{\widehat{N-1}} T_{\widehat{N+1}} = T_{\widehat{N}} T_{\widehat{N-2}} \cdots T_{\hat{2}},
\]
where the complementary indices now refer to $[N+1] \setminus \{1\} = \{2,\ldots,N+1\}$.
By renaming $T_{\hat{k}}$ to $T_{\widehat{k-1}}$ in the last equation, it reads
\[
 T_{\hat{2}} \cdots T_{\widehat{N-2}} T_{\widehat{N}} = T_{\widehat{N-1}} T_{\widehat{N-3}} \cdots T_{\hat{1}},
\]
where now the indices are complementary in $[N]$, so that we have the standard form of the (even) $N$-gon equation.
\end{proof}

\begin{Remark}
The last result means that, if \smash{$\big\{T_K | K \in {[N+1] \choose N} \big\}$} is a solution of an odd $(N+1)$-gon equation,
then each map $T_K$ with $1 \in K$ has the form
\[
 T_K = S_{K'} \times T_{K'}, \qquad K' := K \setminus \{1\},
\]
with a map \smash{$S_{K'} \colon \cU_{\vec{P}_{\rm o}(K)} \rightarrow \cU_{K'}$}, and \smash{$\big\{ T_{K'} | K' \in { \{ 2,3, \ldots, N+1 \} \choose N-1} \big\}$} solve the even $N$-gon equation.
Each solution of the odd $(N+1)$-gon equation is thus an extension, of the above form, of a~solution of the
(even) $N$-gon equation.
\end{Remark}

\begin{Example}
The pentagon equation $T_{\hat{1}} T_{\hat{3}} T_{\hat{5}} = T_{\hat{4}} T_{\hat{2}}$ involves
the maps $T_{\hat{1}} \colon \cU_{234} \times \cU_{245} \rightarrow \cU_{345} \times \cU_{235}$,
$T_{\hat{2}} \colon \cU_{134} \times \cU_{145} \rightarrow \cU_{345} \times \cU_{135}$,
 $T_{\hat{3}} \colon\cU_{124} \times \cU_{145} \rightarrow \cU_{245} \times \cU_{125}$,
$T_{\hat{4}} \colon \cU_{123} \times \cU_{135} \rightarrow \cU_{235} \times \cU_{125}$,
 $T_{\hat{5}} \colon \cU_{123} \times \cU_{134} \rightarrow \cU_{234} \times \cU_{124}$.
We thus have
\[
 T_{\hat{4}} T_{\hat{2}} \colon\ \cU_{123} \times \cU_{134} \times \cU_{145} \rightarrow
(\cU_{345} \times \cU_{235}) \times \cU_{125}
\]
 and
\[
 T_{\hat{3}} T_{\hat{5}}\colon\ \cU_{123} \times \cU_{134} \times \cU_{145} \rightarrow
(\cU_{234} \times \cU_{245}) \times \cU_{125}.
\]
Recalling that $T_{\hat{1}}$ does not involve any of
the spaces $\cU_{1ij}$, $1<i<j \leq 5$, the pentagon equation implies that the maps, obtained from $\{ T_{\hat{k}} \}$
by deleting the first component of their codomain, satisfy the tetragon equation (with indices shifted by $1$).
\end{Example}

Let us now turn to the \emph{dual} $(N+1)$-gon equation and consider the subset of maps
$\tT_K \colon \cU_{\vec{P}_{\rm e}(K)}\allowbreak \rightarrow \cU_{\cev{P}_{\rm o}(K)}$ with $1 \in K$.

\begin{Theorem}
\label{thm:dualpolygon->evendualpolygon_red1}
Let $N \in \bbN$ be even, $N>2$.
\begin{itemize}\itemsep=0pt
\item[$(1)$] Let $\tT_K$, \smash{$K \in {[N+1] \choose N}$}, be maps solving the dual $(N+1)$-gon equation.
For $K$ with $1 \in K$, let~$\tT_K$ not depend on the last component of its domain. Let
$\tT_{K'}$, $K' = K \setminus \{1\}$, be obtained from $\tT_K$ by excluding the last
component of its domain. Then \smash{$\big\{\tT_{K'} | K' \in {\{2,\ldots, N+1 \} \choose N-1 } \big\}$} solve the dual $N$-gon equation.
\item[$(2)$] Each solution of the dual $N$-gon equation can be extended to a degenerate solution of the dual $(N+1)$-gon equation.
 \end{itemize}
\end{Theorem}
\begin{proof}
(1) If $N$ is even, then $K' = K\setminus \{1\}$ is the last element of \smash{$\vec{P}_{\rm e}(K)$}, so that
\[
 \tT_K \colon\ \cU_{\vec{P}'_{\rm e}(K)} \times \cU_{K \setminus \{1\}} \rightarrow \cU_{\cev{P}_{\rm o}(K)},
\]
where \smash{$\vec{P}'_{\rm e}(K)$} is \smash{$\vec{P}_{\rm e}(K)$} without the last element. If $\tT_K$ does not depend on the last
component of its domain, it induces a map
\[
 \tT_{K'} \colon\ \cU_{\vec{P}_{\rm e}(K')} \rightarrow \cU_{\cev{P}_{\rm o}(K')},
 \qquad K' \in {\{2,\ldots, N+1 \} \choose N-1},
\]
where we set $\cU_J := \cU_{\{1\} \cup J}$. With the same argument as in the proof of
Theorem~\ref{thm:polygon->oddpolygon_red1}, the dual~${(N+1)}$-gon equation
\[
 \tT_{\hat{2}} \tT_{\hat{4}} \cdots \tT_{\widehat{N}}
 = \tT_{\widehat{N+1}} \tT_{\widehat{N-1}} \cdots \tT_{\hat{3}} \tT_{\hat{1}}
\]
reduces to
\[
 \tT_{\hat{2}} \tT_{\hat{4}} \cdots \tT_{\widehat{N}}
 = \tT_{\widehat{N+1}} \tT_{\widehat{N-1}} \cdots \tT_{\hat{3}},
\]
where now the complementary indices refer to $[N+1] \setminus \{1\}$. A shift in the numbering achieves
the standard form \eqref{dual_N-gon_even} of the dual even $N$-gon equation.

(2) The proof is analogous to that of part~(2) of Theorem~\ref{thm:polygon->oddpolygon_red1}.
\end{proof}

\begin{Theorem}
\label{thm:dualpolygon->odddualpolygon_red1}
Let $N \in \bbN$ be odd, $N>2$. Let $\tT_K$, \smash{$K \in {[N+1] \choose N}$}, be maps solving the dual $(N+1)$-gon equation.
For $K$ with $1 \in K$ let $\tT_{K'}$, $K' = K \setminus \{1\}$, be obtained from $\tT_K$ by excluding the first
component of its codomain. Then \smash{$\big\{\tT_{K'} | K' \in {\{2,\ldots, N+1 \} \choose N-1 } \big\}$} solve the dual $N$-gon equation.
\end{Theorem}
\begin{proof}
If $N$ is odd and $1 \in K$, then $K'$ is the first element of $\cev{P}_{\rm o}(K)$, so that
 \[
 \tT_K \colon\ \cU_{\vec{P}_{\rm e}(K)} \rightarrow \cU_{K \setminus \{1\}} \times \cU_{\cev{P}'_{\rm o}(K)},
\]
where $\cev{P}'_{\rm o}(K)$ is $\cev{P}_{\rm o}(K)$ without the first element. Now each $\tT_K$, $1 \in K$, induces a map
\[
 \tT_{K'} \colon\ \cU_{\vec{P}_{\rm e}(K')} \rightarrow \cU_{\cev{P}_{\rm o}(K')},
 \qquad K' \in {\{2,\ldots, N+1 \} \choose N-1}.
\]
The dual $(N+1)$-gon equation then reduces to the dual odd $N$-gon equation, using an argument as
in the proofs of the preceding theorems.
\end{proof}

\subsection[Reductions induced by pr\_N+1]{Reductions induced by $\boldsymbol{\mathrm{pr}_{N+1}}$}
The projection $T(N+1,N-1) \rightarrow \tilde{T}(N,N-2)$ induces a reduction of the $(N+1)$-gon to
the dual $N$-gon equation.
It is essentially obtained by dropping the map $T_K$ with $N+1 \notin K$. But we have to arrange that the remaining maps $T_K$, $N+1 \in K$, are reduced to maps between (products of) spaces $U_J$ with $N+1 \in J \in P(K)$.

\begin{Theorem}
\label{thm:polygon->dualpolygon_red2}
Let $N \in \bbN$, $N>2$.
\begin{itemize}\itemsep=0pt
\item[$(1)$] Let $T_K$, $K \in {[N+1] \choose N}$, be maps solving the $(N+1)$-gon equation.
For $K$ with $N+1 \in K$, let $T_K$ not depend on the first component of its domain.
Set $K' = K \setminus \{N+1\}$ and let $\tT_{K'}$ be given by $T_K$ with the first component
of its domain excluded. Then \smash{$\big\{ \tT_{K'} | K' \in {[N] \choose N-1} \big\}$} solve the dual $N$-gon equation.
\item[$(2)$] Each solution of the dual $N$-gon equation can be extended to a degenerate solution of the~${(N+1)}$-gon equation.
 \end{itemize}
\end{Theorem}
\begin{proof}
(1) Let $K$ be such that $N+1 \in K$. Then $K' := K \setminus \{N+1\}$ is the first element of~$P(K)$ in lexicographic order,
and thus also of $P_{\rm o}(K)$. Hence, we have maps
\[
 T_K \colon\ \cU_{K \setminus \{N+1\}} \times \cU_{\vec{P}'_{\rm o}(K)} \rightarrow \cU_{\cev{P}_{\rm e}(K)},
 \qquad K \in {[N+1] \choose N}, \qquad N+1 \in K,
\]
where $P'_{\rm o}(K)$ is $P_{\rm o}(K)$ without its first element (in lexicographic order). If these maps do not
depend on the first component of their domain, they project to maps
\[
 \tT_{K'}\colon\ \cU_{\vec{P}_{\rm e}(K')} \rightarrow \cU_{\cev{P}_{\rm o}(K')}, \qquad K' \in {[N] \choose N-1},
\]
where we set \smash{$\cU_J := \cU_{J \cup \{N+1\}}$} for $J \in P(K')$. Under this assumption,
\[T_{\widehat{N+1}} \colon\
\cU_{\vec{P}_{\rm o}(\{1,2,\ldots, N\})} \rightarrow \cU_{\cev{P}_{\rm e}(\{1,2,\ldots, N\})}
\]
 maps to spaces that
are disregarded by all other $T_K$ appearing in the $(N+1)$-gon equation.

For odd $N$, the $(N+1)$-gon equation
\[
 T_{\hat{2}} T_{\hat{4}} \cdots T_{\widehat{N-1}} T_{\widehat{N+1}}
 = T_{\widehat{N}} T_{\widehat{N-2}} \cdots T_{\hat{1}}
\]
thus reduces to
\[
 \tT_{\hat{2}} \tT_{\hat{4}} \cdots \tT_{\widehat{N-1}}
 = \tT_{\widehat{N}} \tT_{\widehat{N-2}} \cdots \tT_{\hat{3}} \tT_{\hat{1}},
\]
where now the indices are complementary in $[N+1] \setminus \{N+1\} = [N]$.
This is the dual $N$-gon equation for odd $N$.

For even $N$, the $(N+1)$-gon equation
\[
 T_{\hat{1}} T_{\hat{3}} \cdots T_{\widehat{N-1}} T_{\widehat{N+1}}
 = T_{\widehat{N}} T_{\widehat{N-2}} \cdots T_{\hat{2}}
\]
reduces to
\[
 \tT_{\hat{1}} \tT_{\hat{3}} \cdots \tT_{\widehat{N-1}} = \tT_{\widehat{N}} \tT_{\widehat{N-2}} \cdots \tT_{\hat{2}},
\]
where again indices are now complementary in $[N]$. This is the dual $N$-gon equation for even~$N$.

(2) The proof is analogous to that in preceding theorems.
\end{proof}

Let us now turn to the corresponding reduction of the \emph{dual} $(N+1)$-gon equation and consider the maps
\smash{$\tT_K \colon\cU_{\vec{P}_{\rm e}(K)} \rightarrow \cU_{\cev{P}_{\rm o}(K)}$} with $N+1 \in K$.

\begin{Theorem}
\label{thm:dualpolygon->polygon_red2}
Let $N \in \bbN$, $N>2$. Let $\tT_K$, \smash{$K \in {[N+1] \choose N}$}, be maps solving the dual $(N+1)$-gon equation.
For $K$ with $N+1 \in K$, set $K' = K \setminus \{N+1\}$ and let $T_{K'}$ be given by $\tT_K$ with the last component
of its codomain excluded. Then \smash{$\big\{ T_{K'} | K' \in {[N] \choose N-1} \big\}$} solve the $N$-gon equation.
\end{Theorem}
\begin{proof}
 Since $K' = K \setminus \{N+1\}$ is the first element of \smash{$\vec{P}(K)$}, and thus the last element of~\smash{$\cev{P}_{\rm o}(K)$},
\[
 \tT_K \colon\ \cU_{\vec{P}_{\rm e}(K)} \rightarrow \cU_{\cev{P}'_{\rm o}(K)} \times \cU_{K \setminus \{N+1\}},
\]
where $\cev{P}'_{\rm o}(K)$ is $\cev{P}_{\rm o}(K)$ without the last element. Hence, $\tT_K$ induces a map
\[
 T_{K'} \colon\ \cU_{\vec{P}_{\rm o}(K')} \rightarrow \cU_{\cev{P}_{\rm e}(K')},
\]
where we set again \smash{$\cU_J := \cU_{J \cup \{N+1\}}$} for $J \in P(K')$. With the kind of argument used in the
proofs of preceding theorems, the dual $(N+1)$-gon equation reduces to the $N$-gon equation.
\end{proof}

\begin{Remark}
The last result means that, if \smash{$\big\{\tT_K | K \in {[N+1] \choose N} \big\}$} is a solution of a dual $(N+1)$-gon equation,
then each map $\tT_K$ with $N+1 \in K$ has the form
\[
 \tT_K = T_{K'} \times S_{K'}, \qquad K' := K \setminus \{N+1\},
\]
with a map \smash{$S_{K'} \colon \cU_{\vec{P}_{\rm e}(K)} \rightarrow \cU_{K'}$}, and \smash{$\{ T_{K'} | K' \in { [N] \choose N-1} \}$} solve the $N$-gon equation.
Each solution of the dual $(N+1)$-gon equation is thus an extension, of the above form, of a solution of the~$N$-gon equation.
\end{Remark}

\section{Examples of polygon equations}
\label{sec:ExPEs}
In this section, we elaborate polygon equations up to the 8-gon equation.

\subsection{Trigon equation}
This is the equation
\begin{gather}
 T_{23} T_{12} = T_{13} \label{3gon_eq}
\end{gather}
for maps $T_{ij}\colon \cU_i \rightarrow \cU_j$, $i<j $.
On the left-hand side of \eqref{3gon_eq}, we mean the composition of two maps.
If the sets are the same, $\cU_i = \cU$, $i=1,2,3$, and if there is only a single map $T$, then
the trigon equation means that $T$ is idempotent.

\subsection{Tetragon equation}
For $i,j=1,2,3,4$, $i<j$, let a map $L_{ij}\colon \cU_i \rightarrow \cU_j$ carry a parameter from a set $\cU_{ij}$.
Let us further assume that each of the \emph{local trigon equations}
\[
 L_{jk}(u_{jk}) L_{ij}(u_{ij}) = L_{ik}(T_{ijk}(u_{ij},u_{jk})), \qquad i<j<k,
\]
uniquely determines a map
\[
 T_{ijk} \colon\ \cU_{ij} \times \cU_{jk} \longrightarrow \cU_{ik}.
\]
\begin{figure}[t]
\centering
\includegraphics[scale=.45]{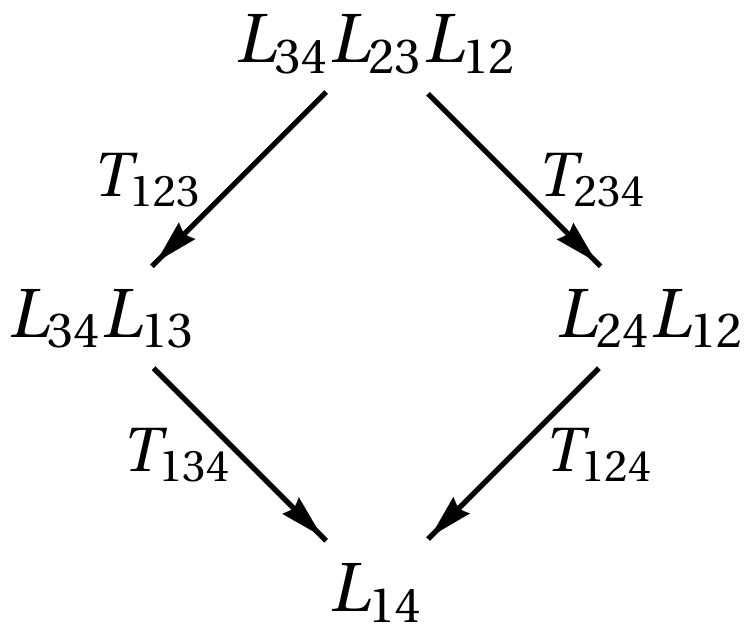}

\caption{From local trigon equations to the tetragon equation. Here, and correspondingly in following
figures, we suppress the parameters of the maps $L$.}\label{fig:tri2tetra}
\end{figure}
Using associativity of compositions, there is a consistency condition, see Figure~\ref{fig:tri2tetra},\footnote{Such
a diagram already appeared
in the 1965 PhD Thesis of James Wirth, see \cite{Wirt+Stas06}, where the maps are, however, of a
different nature, so that the step to higher polygon equations does not work.}
which requires that the maps $T_{ijk}$ have to satisfy the \emph{tetragon equation}
\[
 T_{134} T_{123} = T_{124} T_{234}.
\]
Both sides of this equation act on the lexicographically ordered Cartesian product
$\cU_{12} \times \cU_{23} \times \cU_{34}$ and map to $\cU_{14}$.

Let us introduce a boldface ``position index" that indicates
the first of two neighboring sets on which the respective map acts in a Cartesian product of more
than two sets,
\[
 T_{134} T_{123,\bsy{1}} = T_{124} T_{234,\bsy{2}}.
\]
These additional indices are redundant, as long as we keep the combinatorial indices
and keep track of the numbered sets.

Using complementary index notation, where $\hat{k}$ stands for the complement of $k$ in $\{1,2,3,4\}$,
the tetragon equation reads $
 T_{\hat{2}} T_{\hat{4},\bsy{1}} = T_{\hat{3}} T_{\hat{1},\bsy{2}}$.

Writing
\[
 T_{ijk}(a_{ij},a_{jk}) =: a_{ij} \bullet_{ijk} a_{jk},
\]
with $a_{ij} \in \cU_{ij}$, the equation takes the form
\[
 (a_{12} \bullet_{123} a_{23}) \bullet_{134} a_{34} = a_{12} \bullet_{124} (a_{23} \bullet_{234} a_{34}),
\]
which is a mixed associativity condition for the (in general different) binary operations $\bullet_{123}$, $\bullet_{124}$, $\bullet_{134}$ and $\bullet_{234}$.\footnote{Examples of such associativity
relations for different binary operations are provided, for example, by nonsymmetric Poisson
algebras (in a setting of vector spaces, with $\times$ replaced by the corresponding tensor product)~\cite{Mark96}.}

In the simplest case, where all the basic sets are equal and we are dealing with a single map $T$, we may drop
the combinatorial indices, but retain the boldface ``position" indices. The tetragon equation is then $T T_{\bsy{1}} = T T_{\bsy{2}} $.
Writing $T(a,b) = a \cdot b$, it becomes the associativity relation for the binary operation $\cdot$.

\subsection{Pentagon equation}
For $i,j,k=1, \ldots,5$, $i<j<k$, let a map
\[
 L_{ijk} \colon\ \cU_{ij} \times \cU_{jk} \longrightarrow \cU_{ik}
\]
depend on a parameter from a set $\cU_{ijk}$.
Let us assume that each of the \emph{local tetragon equations}
\[
 L_{ikl}(u_{ikl}) L_{ijk,\bsy{1}}(u_{ijk}) = L_{ijl}(v_{ijl}) L_{jkl,\bsy{2}}(v_{jkl} )
, \qquad 1 \leq i<j<k<l \leq 5,
\]
uniquely determines a map
\[
 T_{ijkl} \colon\ \cU_{ijk} \times \cU_{ikl} \rightarrow \cU_{jkl} \times \cU_{ijl}
\]
via $(u_{ijk}, u_{ikl}) \mapsto (v_{jkl}, v_{ijl})$.
Then it follows that the maps $T_{ijkl}$, $1 \leq i < j < k < l \leq 5$, have to satisfy the pentagon equation,
see Figure~\ref{fig:tetra2pent}.
\begin{figure}[t]\centering
\includegraphics[scale=.70]{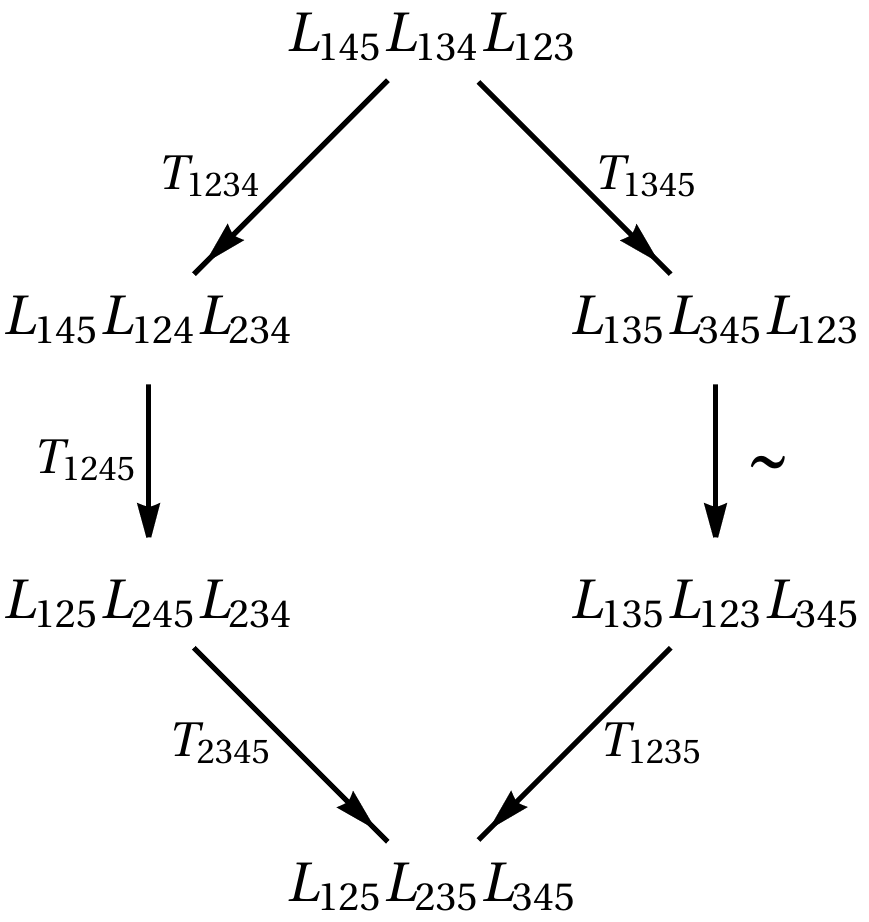}
\hspace{2cm}
\includegraphics[scale=.56]{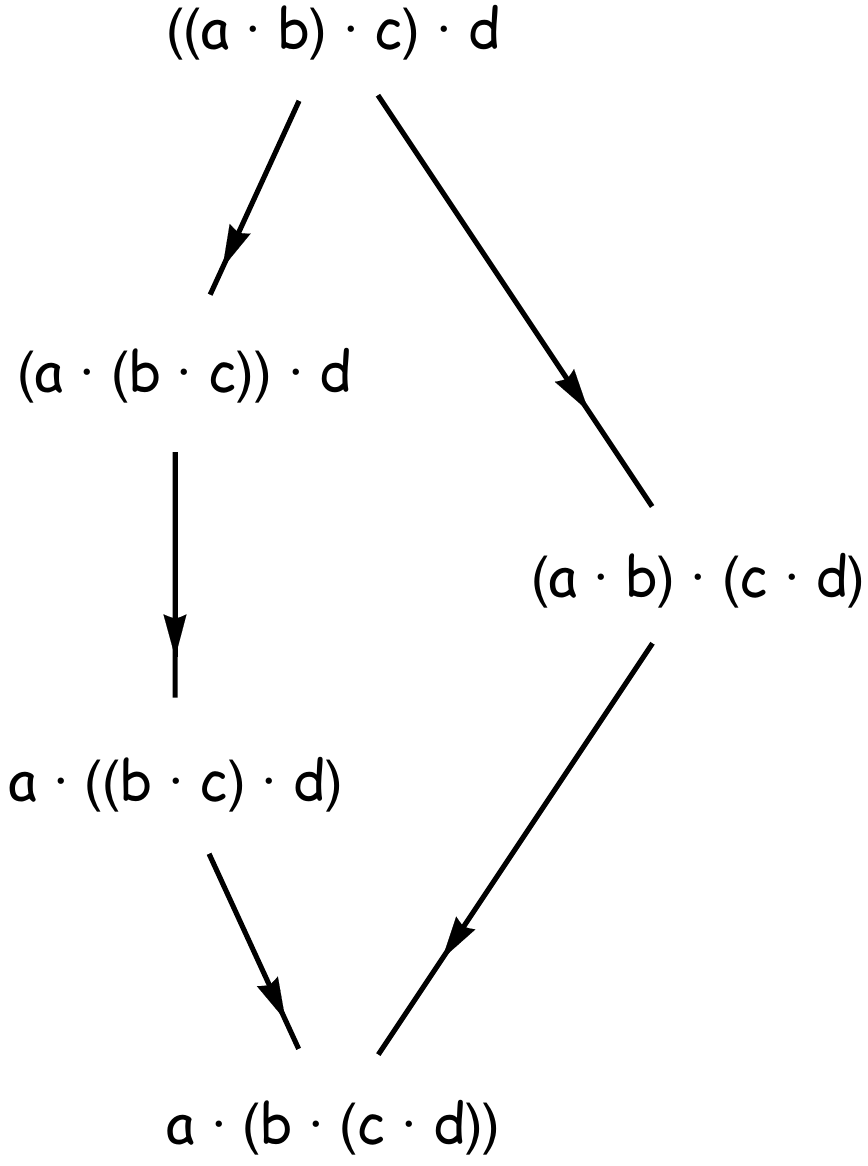}
\caption{From local tetragon equations to the pentagon equation. Here $\boldsymbol{\sim}$
stands for an equivalence, which corresponds to an application of a transposition map $\cP$.
Essentially, the left diagram is the pentagonal \emph{Tamari lattice}, in its original form,
as displayed in the second figure.
Here the action of a map $T$ corresponds to a right associativity map
$(a \cdot b) \cdot c \mapsto a \cdot (b \cdot c)$. If this is an invertible map in a category with a binary operation $\cdot$, the pentagon relation means that it is a~\emph{coherent associativity isomorphism} in the sense of \cite{MacL63}.}\label{fig:tetra2pent}
\end{figure}

Using complementary index notation, the \emph{pentagon equation} is
\[
 T_{\hat{1}} T_{\hat{3}} T_{\hat{5}} = T_{\hat{4}} T_{\hat{2}}.
\]
Both sides of this equation act on
$\cU_{123} \times \cU_{134} \times \cU_{145}$ and map to $\cU_{345} \times \cU_{235} \times \cU_{125}$.

Representing the equivalence relation $\boldsymbol{\sim}$ in the diagram in Figure~\ref{fig:tetra2pent}
by a transposition map, $\cP(a,b)=(b,a)$, and reading off on which neighboring positions in a multiple
Cartesian product a map acts, we get
\[
 T_{\hat{1},\bsy{1}} T_{\hat{3},\bsy{2}} T_{\hat{5},\bsy{1}}
 = T_{\hat{4},\bsy{2}} \cP_{\bsy{1}} T_{\hat{2},\bsy{2}}.
\]

In case of identical basic sets, then renamed to $\cU$, and a single map $T$, the last equation takes the form
\begin{gather}
 T_{\bsy{1}} T_{\bsy{2}} T_{\bsy{1}}
 = T_{\bsy{2}} \cP_{\bsy{1}} T_{\bsy{2}}, \label{5gon_eq_pos}
\end{gather}
where the combinatorial indices have been dropped. Now all the information needed is provided
by the ``position'' indices.
The latter is our abbreviation of the following more familiar form of the pentagon equation
\[
 T_{\bsy{12}} T_{\bsy{23}} T_{\bsy{12}}
 = T_{\bsy{23}} \cP_{\bsy{12}} T_{\bsy{23}}.
\]

Writing
\begin{gather}
 T(a,b) = (a \ast b, a \cdot b), \label{5gon_multiplications}
\end{gather}
the last restricted form of the pentagon equation is equivalent to the conditions (cf.\ \cite{Kash+Resh07,Kash+Serg98})
\begin{gather}
 (a \ast b) \ast ((a \cdot b) \ast c) = b \ast c, \qquad
 (a \ast b) \cdot ((a \cdot b) \ast c) = a \ast (b \cdot c),\nonumber \\
 (a \cdot b) \cdot c = a \cdot (b \cdot c), \label{5gon_prod_conditions}
\end{gather}
for all $a,b,c \in \cU$.
Set-theoretic solutions have been obtained in \cite{Baaj+Skan93,CMM19,CMS20,CJK20,Jian+Liu05,Kash96AA,Kash99TMP,Kash00TMP,Kash11,Kash+Serg98,Kass23,Mazz23,MPS24,Zakr92}.

\begin{Example} \quad
\begin{itemize}\itemsep=0pt
\item[(1)] If $a \ast b = b$ for all $a,b \in \cU$, \eqref{5gon_prod_conditions} reduces to the associativity condition for $\cdot$.
If $(\cU,\cdot)$ is a~group and if $T$ is invertible, then this is the only solution \cite{CJK20,Kash+Serg98}.
It underlies one of the Kac--Takesaki operators on a group (see, e.g., \cite{Timm08}).
If $\cU$ is a subset of a group~$(G,\cdot)$, not containing the identity element,
there are more solutions \cite{Kash+Serg98}.
\item[(2)] If $a \cdot b =a$, the above system reduces to $(a \ast b) \ast (a \ast c) = b \ast c$.
In a group, a solution is given by $a \ast b = a^{-1} b$. This underlies another Kac--Takesaki operator on a group (see, e.g.,~\mbox{\cite{CMM19,Timm08}}).
\end{itemize}
\end{Example}

In terms of the composition $\hat{T} := T\cP$, \eqref{5gon_eq_pos} takes the form
\begin{gather}
 \hat{T}_{\bsy{12}} \hat{T}_{\bsy{13}} \hat{T}_{\bsy{23}}
 = \hat{T}_{\bsy{23}} \hat{T}_{\bsy{12}}. \label{5gon_eq_hat}
\end{gather}

\begin{Remark}
If $T$ solves \eqref{5gon_eq_pos}, then $\check{T} := \cP T$ satisfies the -- relative to \eqref{5gon_eq_hat} -- \emph{reversed} pentagon equation
\begin{gather}
 \check{T}_{\bsy{12}} \check{T}_{\bsy{23}} = \check{T}_{\bsy{23}} \check{T}_{\bsy{13}} \check{T}_{\bsy{12}}.
 \label{5gon_eq_check}
\end{gather}
If an invertible map $\hat{T}$ satisfies \eqref{5gon_eq_hat}, then its inverse $\hat{T}^{-1}$ satisfies the above reversed equation.
\end{Remark}

\begin{Remark}
If $\hat{T}$ is involutive, then it satisfies both, \eqref{5gon_eq_hat} and \eqref{5gon_eq_check}.
Such solutions have been explored in \cite{CJK20}.
\end{Remark}

\subsubsection{An example related to incidence geometry}

Let $\cV$ be a vector space and $L(a) \colon \cV \times \cV \rightarrow \cV$, $a \in \cU$, be such that the ``local''
tetragon equation~${L(b) L(a)_{\boldsymbol{1}} = L(b') L(a')_{\boldsymbol{2}}}$
determines a unique map $T \colon \cU \times \cU \rightarrow \cU \times \cU$ via
${(a,b) \mapsto (a',b')}$. Then this map satisfies the pentagon equation \eqref{5gon_eq_pos}.
Writing $x \circ_a y := L(a)(x,y)$, we can express the above equation as
the parameter-dependent associativity condition $(x \circ_a y) \circ_b z = x \circ_{b'} (y \circ_{a'} z) $.

An example is given by $\cV = \bbR^n$, $\cU = (0,1) \subset \bbR$, and $x \circ_a y := a x + (1-a) y$.
Then we obtain the solution
\begin{gather}
 T(a,b) = \left( \frac{(1 - a) b}{1- a b}, a b \right) \label{5gon_sol_proj_geom}
\end{gather}
of the pentagon equation (also see \cite{Kash99TMP}). It also shows up as a map of ``polarizations''
resulting from the evolution of tree-shaped matrix KP solitons \cite{DMH18TMP}.

The above binary operation can be interpreted as
the collinearity of three points $A,B,A\circ_a B$: $A \circ_a B = a A + (1-a) B$.
The generalized associativity condition can then be viewed as a $(6_2,4_3)$ configuration \cite{Grue09},
which consists of 6 points and 4 lines, where each point is incident with exactly two lines, each line
with exactly 3 points. See Figure~\ref{fig:incidence_geometry}.
A Desargues configuration~$(10_3)$ consists of 10 points and 10 lines, each point (line) incident
with 3 lines (points). It contains~5 configurations of type $(6_2,4_3)$, which thus constitute a pentagon.
Also see \cite{Doli+Serg11} for related considerations.

\begin{figure}[t]
\centering
 \includegraphics[scale=.6]{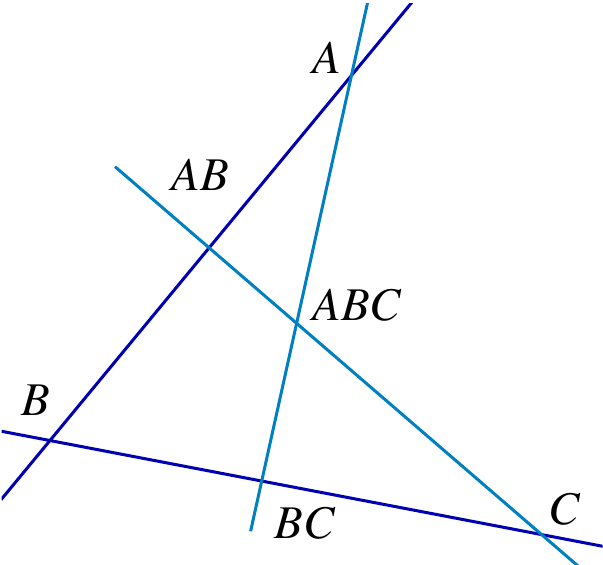}

\caption{A $(6_2,4_3)$ configuration.}\label{fig:incidence_geometry}
\end{figure}

\subsection{Hexagon equation}
For $i,j,k,l=1, \ldots,6$, $i<j<k<l$, let a map
\[
 L_{ijkl} \colon\ \cU_{ijk} \times \cU_{ikl} \longrightarrow \cU_{jkl} \times \cU_{ijl}
\]
depend on a parameter from a set $\cU_{ijkl}$. We assume that each of the \emph{local pentagon equations}
\begin{gather*}
 L_{jklm}(u_{jklm}) L_{ijlm}(u_{ijlm}) L_{ijkl}(u_{ijkl}) = L_{ijkm}(v_{ijkm}) L_{iklm}(v_{iklm}), \\
 1 \leq i<j<k<l<m \leq 6,
\end{gather*}
uniquely determines a map
\[
 T_{ijklm} \colon\ \cU_{ijkl} \times \cU_{ijlm} \times \cU_{jklm} \rightarrow \cU_{iklm} \times \cU_{ijkm}
\]
via $(u_{ijkl}, u_{ijlm}, u_{jklm}) \mapsto (v_{iklm}, v_{ijkm})$.

Then the maps $T_{ijklm}$, $1 \leq i < j < k < l <m \leq 6$, have to satisfy the hexagon equation.
See Figure~\ref{fig:penta2hexa}.
\begin{figure}[t]
\centering
\includegraphics[scale=.5]{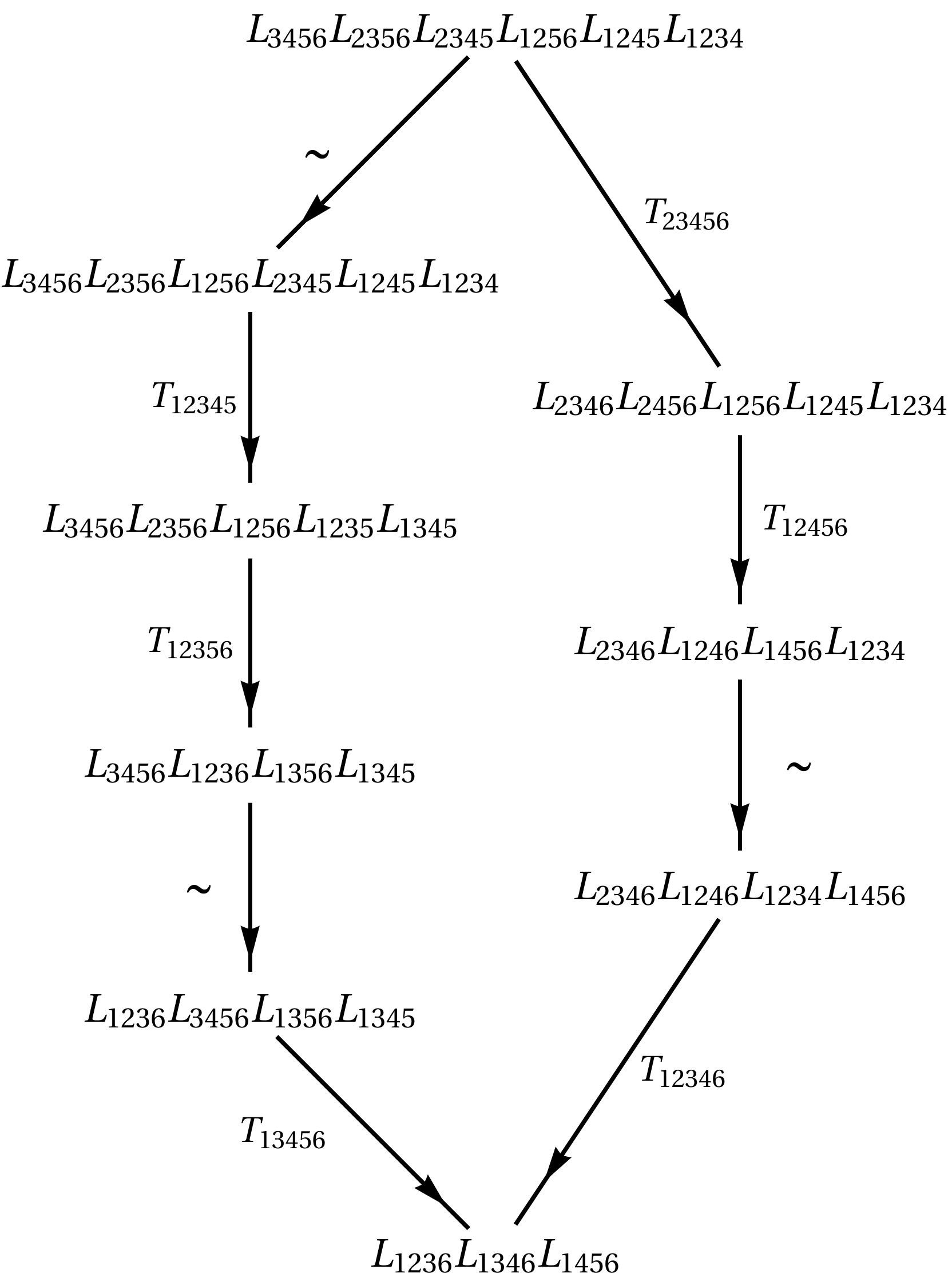}

\caption{From local pentagon equations to the hexagon equation.} \label{fig:penta2hexa}

\end{figure}
Using complementary index notation, the \emph{hexagon equation} reads
\[
 T_{\hat{2}} T_{\hat{4}} T_{\hat{6}} = T_{\hat{5}} T_{\hat{3}} T_{\hat{1}}. \label{6gon_eq}
\]
Both sides act on $\cU_{1234} \times \cU_{1245} \times \cU_{1256} \times \cU_{2345} \times \cU_{2356} \times \cU_{3456}$
and map to $\cU_{1456} \times \cU_{1346} \times \cU_{1236}$.

Employing transposition maps, according to the equivalences $\boldsymbol{\sim}$ appearing in Figure~\ref{fig:penta2hexa},
and introducing position indices, the hexagon equation takes the form
\[
 T_{\hat{2},\bsy{1}} \cP_{\bsy{3}} T_{\hat{4},\bsy{2}} T_{\hat{6},\bsy{1}}
 \cP_{\bsy{3}}
 = T_{\hat{5},\bsy{2}} \cP_{\bsy{1}} T_{\hat{3},\bsy{2}} T_{\hat{1},\bsy{4}}.
\]

If all basic sets are equal and we are dealing with a single map $T$,
we may drop the combinatorial indices and only retain the (boldface) position indices:
\begin{gather}
 T_{\bsy{1}} \cP_{\bsy{3}} T_{\bsy{2}} T_{\bsy{1}}
 \cP_{\bsy{3}}
 = T_{\bsy{2}} \cP_{\bsy{1}} T_{\bsy{2}} T_{\bsy{4}}.
 \label{6-gon_eq_pos}
\end{gather}

\begin{Remark}
According to our knowledge, without introduction of an auxiliary structure as discussed in \cite{DMH15}, the hexagon equation \eqref{6-gon_eq_pos}, and correspondingly all higher \emph{even} polygon equations, cannot be rewritten without the explicit appearance of transpositions. This is in contrast
to the case of the pentagon equation (see \eqref{5gon_eq_hat}) and higher \emph{odd} polygon equations.
\end{Remark}

Expressing $T$ in terms of two ternary operations,
\begin{gather}
 T(a,b,c) = ( \langle a,b,c \rangle, [a,b,c] ), \label{6gon_ternary_ops}
\end{gather}
the hexagon equation, acting on $(a,b,c,d,e,f) \in \cU^6$, is equivalent to the following conditions,
\begin{gather}
 \langle \langle a,b,d \rangle, \langle [a,b,d],c,e \rangle, f \rangle =
 \langle b,c, \langle d,e,f \rangle \rangle, \nonumber \\
 [\langle a,b,d \rangle, \langle [a,b,d],c,e \rangle, f] = \langle a, [b,c,\langle d,e,f \rangle], [d,e,f]\rangle,
 \nonumber \\
 [[a,b,d],c,e] = [a,[b,c,\langle d,e,f \rangle ], [d,e,f] ].
 \label{6gon_ternary_ops_cond}
\end{gather}

\subsection{Heptagon equation}
\label{subsec:7gon}
Let maps
\[
 L_{ijklm} \colon\ \cU_{ijkl} \times \cU_{ijlm} \times \cU_{jklm} \rightarrow \cU_{iklm} \times \cU_{ijkm},
\]
where $1 \leq i <j<k<l<m \leq 7$,
be subject to local hexagon equations
\begin{gather*}
 L_{iklmn}(u_{iklmn}) L_{ijkmn}(u_{ijkmn}) L_{ijklm}(u_{ijklm})\\
 \qquad = L_{ijkln}(v_{ijkln}) L_{ijlmn}(v_{ijlmn}) L_{jklmn}(v_{jklmn}),
\end{gather*}
where $1 \leq i <j<k<l<m<n \leq 7$.
If these equations uniquely determine maps
\[
 T_{ijklmn} \colon\ \cU_{ijklm} \times \cU_{ijkmn} \times \cU_{iklmn}
 \rightarrow \cU_{jklmn} \times \cU_{ijlmn} \times \cU_{ijkln}
\]
via $(u_{ijklm}, u_{ijkmn}, u_{iklmn}) \mapsto (v_{jklmn}, v_{ijlmn}, v_{ijkln})$,
elaborating
\[
 L_{14567} L_{13467} L_{12367} L_{13456} L_{12356} L_{12345}
 = L_{\widehat{23}} L_{\widehat{25}} L_{\widehat{45}} L_{\widehat{27}} L_{\widehat{47}} L_{\widehat{67}}
\]
in two different ways, using local hexagon equations, we find that the maps $T_{ijklmn}$ have to satisfy
the heptagon equation, which is
\[
 T_{\hat{1}} T_{\hat{3}} T_{\hat{5}} T_{\hat{7}}
 = T_{\hat{6}} T_{\hat{4}} T_{\hat{2}},
 \label{7-gon_eq}
\]
in complementary index notation.

Introducing position indices and transposition maps, we can express it as
\[
 T_{\hat{1},\bsy{1}} T_{\hat{3},\bsy{3}} \cP_{\bsy{5}} \cP_{\bsy{2}}
 T_{\hat{5},\bsy{3}} T_{\hat{7},\bsy{1}} \cP_{\bsy{3}}
 = \cP_{\bsy{3}} T_{\hat{6},\bsy{4}} \cP_{\bsy{3}} \cP_{\bsy{2}} \cP_{\bsy{1}}
 T_{\hat{4},\bsy{3}} \cP_{\bsy{2}} \cP_{\bsy{3}}
 T_{\hat{2},\bsy{4}}.
\]
Both sides act on $\cU_{12345} \times \cU_{12356} \times \cU_{12367} \times \cU_{13456} \times \cU_{13467} \times \cU_{14567}$ and map to
$\cU_{34567} \times \cU_{23567} \times \cU_{23457} \times \cU_{12567} \times \cU_{12457} \times \cU_{12347}$
(which can be read off from \cite[Figure~18]{DMH15}).

If all basic sets are equal and there is only a single map $T$, all the information we need is in the position indices, so that the combinatorial indices can be dropped,
\[
 T_{\bsy{1}} T_{\bsy{3}} \cP_{\bsy{5}} \cP_{\bsy{2}}
 T_{\bsy{3}} T_{\bsy{1}} \cP_{\bsy{3}}
 = \cP_{\bsy{3}} T_{\bsy{4}} \cP_{\bsy{3}} \cP_{\bsy{2}} \cP_{\bsy{1}}
 T_{\bsy{3}} \cP_{\bsy{2}} \cP_{\bsy{3}}
 T_{\bsy{4}}.
\]

Expressing $T$ in terms of three ternary operations,\footnote{It should be clear from the context when $\{\,,\,, \, \}$
denotes a ternary operation and not a set.}
\begin{gather}
 T(a,b,c) = ( \{ a,b,c \}, \langle a,b,c \rangle, [a,b,c] ), \label{7gon_ternary_ops}
\end{gather}
the heptagon equation, evaluated on $(a,b,c,d,e,f)\! \in\! \cU^6$, is equivalent to \eqref{6gon_ternary_ops_cond}, supplemented~by
\begin{gather}
 \{ \{ a,b,d \}, \{ [a,b,d],c,e \}, \{ \langle a,b,d \rangle, \langle [a,b,d], c,e \rangle, f \} \}
 = \{ d, e, f \}, \nonumber \\
 \langle \{ a,b,d \}, \{ [a,b,d], c, e \}, \{ \langle a,b,d \rangle, \langle [a,b,d], c, e \rangle, f \} \rangle
 = \{ b,c, \langle d,e,f \rangle \}, \nonumber \\
 [ \{a,b,d\}, \{[a,b,d],c,e\}, \{ \langle a,b,d \rangle, \langle [a,b,d],c,e \rangle, f \} ]
 = \{ a, [b,c,\langle d,e,f \rangle ],[d,e,f] \}.
 \label{7gon_conditions_beyond_6gon}
\end{gather}

\subsection{Octagon equation}
We consider maps
\[
 T_{ijklmpq}\colon\ \cU_{ijklmp} \times \cU_{ijklpq} \times \cU_{ijlmpq} \times \cU_{jklmpq}
 \rightarrow \cU_{iklmpq} \times \cU_{ijkmpq} \times \cU_{ijklmp},
\]
where $i<j<k<l<m<p<q$.
The \emph{octagon equation} is
\[
 T_{\hat{2}} T_{\hat{4}} T_{\hat{6}} T_{\hat{8}} = T_{\hat{7}} T_{\hat{5}} T_{\hat{3}} T_{\hat{1}}.
\]
It arises as the consistency condition of a system of local heptagon equations.

Using position indices and transposition maps, it can be written as
\begin{gather*}
 T_{\hat{2},\bsy{1}} \cP_{\bsy{4}} \cP_{\bsy{5}} \cP_{\bsy{6}}
 T_{\hat{4},\bsy{3}} \cP_{\bsy{6}} \cP_{\bsy{5}} \cP_{\bsy{2}}
 T_{\hat{6},\bsy{3}} \cP_{\bsy{6}} T_{\hat{8},\bsy{1}} \cP_{\bsy{4}}
 \cP_{\bsy{5}} \cP_{\bsy{6}} \cP_{\bsy{3}} \\
 \qquad= \cP_{\bsy{3}} T_{\hat{7},\bsy{4}} \cP_{\bsy{3}} \cP_{\bsy{2}} \cP_{\bsy{1}}
 T_{\hat{5},\bsy{3}} \cP_{\bsy{6}} \cP_{\bsy{2}} \cP_{\bsy{3}}
 T_{\hat{3},\bsy{4}} T_{\hat{1},\bsy{7}}.
\end{gather*}
The two sides act on $\cU_{\widehat{78}} \times \cU_{\widehat{58}} \times \cU_{\widehat{56}} \times \cU_{\widehat{38}}
\times \cU_{\widehat{36}} \times \cU_{\widehat{34}} \times \cU_{\widehat{18}} \times \cU_{\widehat{16}}
\times \cU_{\widehat{14}} \times \cU_{\widehat{12}}$ and map to
$\cU_{\widehat{23}} \times \cU_{\widehat{25}} \times \cU_{\widehat{27}} \times \cU_{\widehat{45}}
\times \cU_{\widehat{47}} \times \cU_{\widehat{67}}$, which can be read off from \cite[Figure~19]{DMH15}.

If all basic sets are equal and we are dealing with a single map $T$, this reduces to
\begin{gather*} 
 T_{\bsy{1}} \cP_{\bsy{4}} \cP_{\bsy{5}} \cP_{\bsy{6}}
 T_{\bsy{3}} \cP_{\bsy{6}} \cP_{\bsy{5}} \cP_{\bsy{2}}
 T_{\bsy{3}} \cP_{\bsy{6}} T_{\bsy{1}} \cP_{\bsy{4}}
 \cP_{\bsy{5}} \cP_{\bsy{6}} \cP_{\bsy{3}} = \cP_{\bsy{3}} T_{\bsy{4}} \cP_{\bsy{3}} \cP_{\bsy{2}} \cP_{\bsy{1}}
 T_{\bsy{3}} \cP_{\bsy{6}} \cP_{\bsy{2}} \cP_{\bsy{3}}
 T_{\bsy{4}} T_{\bsy{7}}.
\end{gather*}
Writing
\[
 T(a,b,c,d) = ( \{ a,b,c,d \}, \langle a,b,c,d \rangle, [a,b,c,d] ),
\]
with three quaternary operations $\cU^4 \!\rightarrow\! \cU$, the octagon equation, acting
on $(a,b,c,d,e,f,g,h,k,l)$ results in the following six conditions,
\begin{gather}
 \{ \{ a, b, d, g \}, \{ [a, b, d, g], c, e, h \}, \{ \langle a, b, d, g \rangle, \langle [a, b, d, g], c, e, h\rangle, f, k \}, l \}\nonumber\\
 \qquad = \{ d,e,f,\{ g,h,k,l \} \}, \nonumber \\
 \langle \{ a, b, d, g \}, \{ [a, b, d, g], c, e, h \}, \{ \langle a, b, d, g \rangle, \langle [a, b, d, g], c, e, h \rangle, f, k \}, l \rangle
 \nonumber \\
 \qquad
 = \{ b, c, \langle d, e, f, \{g, h, k, l\} \rangle, \langle g, h, k, l \rangle \}, \nonumber \\
 [ \{ a, b, d, g \}, \{ [a, b, d, g], c, e, h \}, \{ \langle a, b, d, g \rangle, \langle [a, b, d, g], c, e, h \rangle, f, k \}, l ] \nonumber \\
 \qquad
 = \{ a, [b, c,\langle d, e, f, \{ g, h, k, l \} \rangle, \langle g, h, k, l \rangle ], [d, e, f, \{g, h, k, l \} ], [g, h, k, l] \}, \nonumber \\
\langle \langle a, b, d, g \rangle, \langle [a, b, d, g], c, e, h \rangle, f, k \rangle
 = \langle b, c, \langle d, e, f, \{ g, h, k, l \} \rangle, [g, h, k, l] \rangle, \nonumber \\
 [ \langle a, b, d, g \rangle, \langle [a, b, d, g], c, e, h \rangle, f, k] \nonumber \\
 \qquad = \langle a, [b, c, \langle d, e, f, \{ g, h, k, l \} \rangle, \langle g, h, k, l \rangle ], [d, e, f, \{ g, h, k, l \} ], [g, h, k, l] \rangle
, \nonumber \\
 [ a, [b, c, \langle d, e, f, \{ g, h, k, l \} \rangle, \langle g, h, k, l \rangle ], [d, e, f, \{ g, h, k, l \} ], [g, h, k, l] ]\nonumber\\
 \qquad = [ [a, b, d, g], c, e, h ], \label{8gon_conditions}
\end{gather}
for all $a,b,c,d,e,f,g,h,k,l \in \cU$.

\section{Examples of dual polygon equations}
\label{sec:dualPol}
In this section, we elaborate dual polygon equations up to the dual 8-gon equation.

\subsection{Dual trigon equation}
This is the equation
\[
 T_{12} T_{23} = T_{13}
\]
for maps $\tT_{ij} \colon \cU_j \rightarrow \cU_i$, $i<j $.

\subsection{Dual tetragon equation}
\label{subsec:dual4gon}

For $i,j=1,2,3,4$, $i<j$, let $L_{ij}\colon \cU_i \rightarrow \cU_j$ carry a parameter from a set $\cU_{ij}$.
Let each of the local trigon equations
\[
 L_{ik}(u_{ik}) = L_{jk}(v_{jk}) L_{ij}(v_{ij}), \qquad i<j<k,
\]
uniquely determine a map
\[
 \tT_{ijk}\colon\ \cU_{ik} \longrightarrow \cU_{ij} \times \cU_{jk}, \qquad u_{ik} \mapsto (v_{ij}, v_{jk}).
\]
Then the maps $\tT_{ijk}$ have to satisfy a consistency condition, which is obtained by reversing
the arrows in Figure~\ref{fig:tri2tetra}.
This means that the maps $\tT_{ijk}$ have to satisfy the dual tetragon equation,{\samepage
\begin{gather}
 \tT_{123,\bsy{2}} \tT_{134} = \tT_{234,\bsy{1}} \tT_{124}, \label{dual4gon_eq}
\end{gather}
which acts on $\cU_{14}$.}

If all basic sets are equal and there is only a single map $\tT$, \eqref{dual4gon_eq} reduces to
\begin{gather}
 \tT_{\bsy{2}} \tT = \tT_{\bsy{1}} \tT. \label{dual4gon_eq_pos}
\end{gather}
Writing $\tT = (\tT_1, \tT_2)$, this amounts to idempotency and commutativity
of the maps $\tT_i\colon \cU \rightarrow \cU$, $i=1,2$.

\begin{Remark}
In a framework of vector spaces, with the Cartesian product replaced by the corresponding tensor product,
and linear maps, \eqref{dual4gon_eq_pos} means that $\tT$ is \emph{coassociative}.
\end{Remark}

\subsection{Dual pentagon equation}

Using complementary index notation, the \emph{dual pentagon equation} is
\[
 \tT_{\hat{5}} \tT_{\hat{3}} \tT_{\hat{1}} = \tT_{\hat{2}} \tT_{\hat{4}},
\]
for maps
\[
 \tT_{ijkl} \colon\ \cU_{ijl} \times \cU_{jkl} \rightarrow \cU_{ikl} \times \cU_{ijk}, \qquad
 1 \leq i<j<k<l \leq 5.
\]
Letting it act on $\cU_{125} \times \cU_{235} \times \cU_{345}$,
this equation takes the form
\[
 \tT_{\hat{5},\bsy{2}} \tT_{\hat{3},\bsy{1}} \tT_{\hat{1},\bsy{2}}
 = \tT_{\hat{2},\bsy{1}} \cP_{\bsy{2}} \tT_{\hat{4},\bsy{1}}.
\]

If all the basic sets are the same and there is only a single map $\tT$, this is simply
\[
 \tT_{\bsy{2}} \tT_{\bsy{1}} \tT_{\bsy{2}} = \tT_{\bsy{1}} \cP_{\bsy{2}} \tT_{\bsy{1}}.
\]
Writing $\tT(a,b) =: (a \cdot b, a \ast b)$, the last equation is equivalent to
\begin{gather}
 a \cdot (b \cdot c) = (a \cdot b) \cdot c, \qquad
 (a \ast (b \cdot c) ) \cdot (b \ast c) = (a \cdot b) \ast c, \nonumber\\
 (a \ast (b \cdot c)) \ast (b \ast c) = a \ast b, \label{dual5gon_conditions}
\end{gather}
for all $a,b,c \in \cU$.

\subsection{Dual hexagon equation}

This is the equation
\[ 
 \tT_{\hat{6}} \tT_{\hat{4}} \tT_{\hat{2}} = \tT_{\hat{1}} \tT_{\hat{3}} \tT_{\hat{5}}
\]
for maps
\[
 \tT_{ijklm} \colon\ \cU_{ijkm} \times \cU_{iklm}
 \rightarrow \cU_{jklm} \times \cU_{ijlm} \times \cU_{ijkl}, \qquad
 1 \leq i<j<k<l < m \leq 6.
\]
Introducing position indices and transposition maps, it takes the form
\[ 
 \tT_{\hat{6},\bsy{1}} \tT_{\hat{4},\bsy{2}} \cP_{\bsy{3}} \tT_{\hat{2},\bsy{1}}
 = \cP_{\bsy{3}} \tT_{\hat{1},\bsy{4}} \tT_{\hat{3},\bsy{2}} \cP_{\bsy{1}} \tT_{\hat{5},\bsy{2}}.
\]

If all the basic sets are the same and there is only a single map $\tT$, this is
\begin{gather}
 \tT_{\bsy{1}} \tT_{\bsy{2}} \cP_{\bsy{3}} \tT_{\bsy{1}}
 = \cP_{\bsy{3}} \tT_{\bsy{4}} \tT_{\bsy{2}} \cP_{\bsy{1}} \tT_{\bsy{2}},
 \label{dual6gon_pos}
\end{gather}
which already appeared as a 4-cocycle condition in \cite{Stre98}.
Writing
\[
 \tT(a,b) =: (a \ast b, a \cdot b, a \diamond b),
\]
\eqref{dual6gon_pos} imposes \eqref{5gon_prod_conditions} on the first two binary operations, and requires in addition
\begin{gather}
 (a \diamond (b \cdot c)) \ast (b \diamond c) = (a \ast b) \diamond ((a \cdot b) \ast c),
 \qquad
 (a \diamond (b \cdot c)) \cdot (b \diamond c) = (a \cdot b) \diamond c, \nonumber \\
 (a \diamond (b \cdot c)) \diamond (b \diamond c) = a \diamond b,
 \label{dualhex_conditions_beyond_pent}
\end{gather}
for all $a,b,c \in \cU$.

\begin{Remark}
\label{rem:Kashaev_Pachner4d}
In \cite{Kash15} (see (3.8) therein), the dual hexagon equation appeared, as a realization of a Pachner move of
type (3,3) in four dimensions, in the form
\[
 (Q \cP)_{\bsy{1}} Q_{\bsy{2}} \cP_{\bsy{1}} Q_{\bsy{2}}
 = \cP_{\bsy{3}} (Q\cP)_{\bsy{4}} Q_{\bsy{2}} \cP_{\bsy{3}} Q_{\bsy{1}}.
\]
Indeed, setting $Q = \tT \cP$, this reads
\[
 \tT_{\bsy{1}} \tT_{\bsy{2}} \cP_{\bsy{2}} \cP_{\bsy{1}} \tT_{\bsy{2}} \cP_{\bsy{2}}
 = \cP_{\bsy{3}} \tT_{\bsy{4}} \tT_{\bsy{2}} \cP_{\bsy{2}} \cP_{\bsy{3}} \tT_{\bsy{1}} \cP_{\bsy{1}}.
\]
Writing it as
\[
 \tT_{\bsy{1}} \tT_{\bsy{2}} \cP_{\bsy{3}} (\cP_{\bsy{3}} \cP_{\bsy{2}} \cP_{\bsy{1}}) \tT_{\bsy{2}} \cP_{\bsy{2}}
 = \cP_{\bsy{3}} \tT_{\bsy{4}} \tT_{\bsy{2}} \cP_{\bsy{1}} (\cP_{\bsy{1}} \cP_{\bsy{2}} \cP_{\bsy{3}}) \tT_{\bsy{1}} \cP_{\bsy{1}},
\]
and using the identities $\cP_{\bsy{3}} \cP_{\bsy{2}} \cP_{\bsy{1}} \tT_{\bsy{2}} = \tT_{\bsy{1}} \cP_{\bsy{2}} \cP_{\bsy{1}}$,
$ \cP_{\bsy{1}} \cP_{\bsy{2}} \cP_{\bsy{3}} \tT_{\bsy{1}} = \tT_{\bsy{2}} \cP_{\bsy{1}} \cP_{\bsy{2}} $,
we obtain
\[
 \tT_{\bsy{1}} \tT_{\bsy{2}} \cP_{\bsy{3}} \tT_{\bsy{1}} \cP_{\bsy{2}} \cP_{\bsy{1}} \cP_{\bsy{2}}
 = \cP_{\bsy{3}} \tT_{\bsy{4}} \tT_{\bsy{2}} \cP_{\bsy{1}} \tT_{\bsy{2}} \cP_{\bsy{1}} \cP_{\bsy{2}} \cP_{\bsy{1}},
\]
which, by use of the braid equation $\cP_{\bsy{2}} \cP_{\bsy{1}} \cP_{\bsy{2}} = \cP_{\bsy{1}} \cP_{\bsy{2}} \cP_{\bsy{1}}$,
is equivalent to \eqref{dual6gon_pos}.
\end{Remark}

\subsection{Dual heptagon equation}
The \emph{dual heptagon equation} is
\begin{gather}
 \cP_{\bsy{3}} \tT_{\hat{7},\bsy{4}} \tT_{\hat{5},\bsy{2}} \cP_{\bsy{4}} \cP_{\bsy{1}}
 \tT_{\hat{3},\bsy{2}} \tT_{\hat{1},\bsy{4}}
 = \tT_{\hat{2},\bsy{1}} \cP_{\bsy{3}} \cP_{\bsy{4}} \tT_{\hat{4},\bsy{2}} \cP_{\bsy{5}} \cP_{\bsy{4}}
 \cP_{\bsy{3}} \tT_{\hat{6},\bsy{1}} \cP_{\bsy{3}},
 \label{dual7-gon_eq}
\end{gather}
with maps $\tT_{ijklmp} \colon \cU_{ijklp} \times \cU_{ijlmp} \times \cU_{jklmp}
 \rightarrow \cU_{iklmp} \times \cU_{ijkmp} \times \cU_{ijklm}$, $i<j<k<l<m<p$.
Both sides act on $\cU_{12347} \times \cU_{12457} \times \cU_{12567} \times \cU_{23457} \times \cU_{23567} \times \cU_{34567}$ and map to $\cU_{14567} \times \cU_{13467} \times \cU_{13456} \times \cU_{12367} \times \cU_{12356} \times \cU_{12345}$.

If all basic sets are the same and there is only a single map $\tT$, writing
\[
 \tT(a,b,c) = (\{ a,b,c \}, \langle a,b,c \rangle, [a,b,c] ),
\]
 \eqref{dual7-gon_eq} is equivalent
to the following conditions for the three ternary operations,
\begin{gather}
 \{ b,c,\{d,e,f\}\} =\{\{a,b,d\}, \{\langle a,b,d\rangle, c,e\},f\}, \nonumber \\
 \{ a, \langle b,c,\{d,e,f\}\rangle, \langle d,e,f\rangle \}
 = \langle \{a,b,d\},\{\langle a,b,d\rangle,c,e\},f\rangle, \nonumber \\
 \langle a, \langle b,c, \{d,e,f\} \rangle, \langle d,e,f \rangle \rangle
 = \langle \langle a,b,d\rangle, c,e \rangle, \nonumber \\
 \{ [a, \langle b,c,\{d,e,f\} \rangle, \langle d,e,f\rangle], [b,c,\{d,e,f\}],[d,e,f]\}
 = [\{a,b,d\}, \{\langle a,b,d \rangle,c,e\},f], \nonumber \\
 \langle [a, \langle b,c,\{d,e,f\} \rangle, \langle d,e,f\rangle], [b,c,\{d,e,f\}],[d,e,f] \rangle
 = [\langle a,b,d \rangle, c,e], \nonumber \\
 [ [a, \langle b,c,\{d,e,f\} \rangle, \langle d,e,f\rangle], [b,c,\{d,e,f\}],[d,e,f] ]
 = [a,b,d],
 \label{dualhept_conditions}
\end{gather}
for all $a,b,c,d,e,f \in \cU$.

\subsection{Dual octagon equation}
\label{subsec:dual_oct}
The \emph{dual octagon equation} is
\[
 \cP_{\bsy{4}} \cP_{\bsy{7}} \cP_{\bsy{5}} \cP_{\bsy{6}} \tT_{\hat{8},\bsy{7}}
 \cP_{\bsy{3}} \tT_{\hat{6},\bsy{4}} \cP_{\bsy{6}} \cP_{\bsy{3}} \cP_{\bsy{2}}
 \tT_{\hat{4},\bsy{3}} \cP_{\bsy{1}} \cP_{\bsy{2}} \cP_{\bsy{3}} \tT_{\hat{2},\bsy{4}}
 = \tT_{\hat{1},\bsy{1}} \tT_{\hat{3},\bsy{3}} \cP_{\bsy{5}} \cP_{\bsy{6}} \cP_{\bsy{2}}
 \tT_{\hat{5},\bsy{3}} \cP_{\bsy{6}} \cP_{\bsy{5}} \cP_{\bsy{4}} \tT_{\hat{7},\bsy{1}}
 \cP_{\bsy{3}},
\]
for maps $\tT_{ijklmpq} \colon \cU_{ijklmp} \times \cU_{ijkmpq} \times \cU_{iklmpq} \rightarrow
\cU_{jklmpq}\times \cU_{ijlmpq} \times \cU_{ijklpq} \times \cU_{ijklmp}$, where $i<j<k<l<m<p<q$.
Both sides of the equation act on
$\cU_{\widehat{67}} \times \cU_{\widehat{47}}\times \cU_{\widehat{45}} \times \cU_{\widehat{27}}
\times \cU_{\widehat{25}} \times \cU_{\widehat{23}}$.

If all basic sets are the same and there is only a single map $\tT$, writing
\[
 \tT(a,b,c) = (\{ a,b,c \}, \langle a,b,c \rangle, [a,b,c], |a,b,c| ),
\]
 with four ternary operations $\cU^3 \rightarrow \cU$, the dual octagon equation, acting
on $(a,b,c,d,e,f)$ results in \eqref{6gon_ternary_ops_cond}, \eqref{7gon_conditions_beyond_6gon},
 and the following conditions,
\begin{gather}
 \{ | a, [b, c, \langle d, e, f \rangle ], [d, e, f] |, | b, c, \langle d, e, f \rangle |, | d, e, f | \} \nonumber \\
\qquad
 = | \{ a, b, d \}, \{ [a, b, d], c, e \}, \{ \langle a, b, d \rangle, \langle [a, b, d], c, e \} \rangle, f |, \nonumber \\
 \langle | a, [b, c, \langle d, e, f \rangle ], [d, e, f] |, | b, c, \langle d, e, f \rangle |, | d, e, f | \rangle
 = | \langle a, b, d \rangle, \langle [a,b.d], c,e, \rangle, f |, \nonumber \\
 [ | a, [b, c, \langle d, e, f \rangle ], [d, e, f] |, | b, c, \langle d, e, f \rangle |, | d, e, f| ] = | [a, b, d], c, e |
, \nonumber \\
 | | a, [b, c, \langle d, e, f \rangle ], [d, e, f] |, | b, c, \langle d, e, f \rangle |, |d, e, f| |
 = | a, b, d |. \label{dual8gon_beyond_hept}
\end{gather}

\section{Relations between solutions of neighboring polygon equations}
\label{sec:relations_polygon_eqs}

By a (dual) polygon map we mean a solution of a (dual) polygon equation.
In this section, the case when all the basic sets appearing in a multiple Cartesian product,
on which a (dual) polygon map acts, are the same set $\cU$ is considered.

First, we formulate a few special cases of theorems in Section~\ref{sec:reductions}. Though they are rather
corollaries, because of their relevance they also deserve to be called theorems. Let $N \in \mathbb{N}$, $N>2$.

\begin{Theorem}
Let $\tT^{(N+1)}$ be a dual $(N+1)$-gon map and $T$ that map with the last component of
its codomain cut off. Then $T$ is an $N$-gon map.
\end{Theorem}

\begin{Theorem}
For even $N$, let $T^{(N+1)}$ be an $(N+1)$-gon map and $T$ that map with the first component of
its codomain cut off. Then $T$ is an $N$-gon map.
\end{Theorem}

\begin{Theorem}
For odd $N$, let $\tT^{(N+1)}$ be a dual $(N+1)$-gon map and $\tT$ that map with the first component of
its codomain cut off. Then $\tT$ is a dual $N$-gon map.
\end{Theorem}

Further theorems derived from those in Section~\ref{sec:reductions} will be formulated in
the following subsections, where we also provide examples for these results and derive more powerful results
for the (dual) polygon equations up to the (dual) octagon equation.

\subsection[Degenerate dual (N+1)-gon maps from (dual) N-gon maps]{Degenerate dual $\boldsymbol{(N+1)}$-gon maps from (dual) $\boldsymbol{N}$-gon maps}

For any (not necessarily trigon or dual trigon) map $T\colon \cU \rightarrow \cU$, $T^{(4)}(a,b) := T(a)$
and also $T^{(4)}(a,b) := T(b)$ are (degenerate) tetragon maps. Furthermore, we have the following.

\begin{Proposition}
Let $\tT^{(4)}$ be a dual tetragon map. Then $T^{(5)}(a,b) := \tT^{(4)}(b)$
is a pentagon map and $\tT^{(5)}(a,b) := \tT^{(4)}(a)$
is a dual pentagon map.
\end{Proposition}
\begin{proof}
This is quickly verified.
\end{proof}

\begin{Proposition}
Let $T^{(5)}$ be a pentagon map. Then $
 T^{(6)}(a,b,c) := T^{(5)}(a,b)$
is a hexagon map.
\end{Proposition}
\begin{proof}
If $\langle a,b,c \rangle = a \ast b$ and $[a,b,c] = a \cdot b$, \eqref{6gon_ternary_ops_cond} becomes
\eqref{5gon_prod_conditions}.
\end{proof}

\begin{Proposition}
Let $\tT^{(5)}$ be a dual pentagon map. Then
$T^{(6)}(a,b,c) := \tT^{(5)}(b,c)$
is a hexagon map.
\end{Proposition}
\begin{proof}
If $\langle a,b,c \rangle = b \cdot c$ and $[a,b,c] =b \ast c$, then \eqref{6gon_ternary_ops_cond} becomes \eqref{dual5gon_conditions}.
\end{proof}

\begin{Proposition}
Let $\tT^{(6)}$ be a dual hexagon map. Then
$T^{(7)}(a,b,c) := \tT^{(6)}(b,c)$
is a heptagon map.
\end{Proposition}
\begin{proof}
Writing $T^{(7)}(a,b,c) = (b \ast c, b \cdot c, b \diamond c)$,
the first two of equations \eqref{7gon_conditions_beyond_6gon} become the first two of \eqref{5gon_prod_conditions}.
The last two of \eqref{6gon_ternary_ops_cond} become the last two of \eqref{dualhex_conditions_beyond_pent}.
The first of equations \eqref{6gon_ternary_ops_cond} becomes the third of \eqref{5gon_prod_conditions},
the third of \eqref{7gon_conditions_beyond_6gon} the first of \eqref{dualhex_conditions_beyond_pent}.
\end{proof}

\begin{Proposition}
Let $\tT^{(6)}$ be a dual hexagon map. Then
$\tT^{(7)}(a,b,c) := \tT^{(6)}(a,b)$
is a dual heptagon map.
\end{Proposition}
\begin{proof}
Writing $\tT^{(7)}(a,b,c) = (a \ast b, a \cdot b, a \diamond b)$,
\eqref{dualhept_conditions} becomes \eqref{5gon_prod_conditions} and
\eqref{dualhex_conditions_beyond_pent}.
\end{proof}

\begin{Proposition}
Let $T^{(7)}$ be a heptagon map. Then
$T^{(8)}(a,b,c,d) := T^{(7)}(a,b,c)$
is an octagon map.
\end{Proposition}
\begin{proof}
If $T$ does not depend on the last argument, setting
\[
 \{ a,b,c,d \} =: \{ a,b,c \}, \qquad
 \langle a,b,c,d \rangle =: \langle a,b,c \rangle, \qquad
 [a,b,c,d] =: [a,b,c],
\]
the first three of conditions \eqref{8gon_conditions} reduce to \eqref{7gon_conditions_beyond_6gon} and the
last three to \eqref{6gon_ternary_ops_cond}. The resulting six conditions are those for $T^{(7)}$ to be a
heptagon map.
\end{proof}

These results suggest that, more generally, the following statements hold, which determine extensions of each solution of an odd $N$-gon equation, or each solution of a dual $N$-gon equation, to a degenerate solution of a (dual) $(N+1)$-gon equation. Indeed, these general results follow from theorems in Section~\ref{sec:reductions}.

\begin{Theorem}
For $n \in \bbN$, $n>1$, let $T^{(2n-1)}$ be a $(2n-1)$-gon map. Then
\[
 T^{(2n)}(a_1,\ldots,a_n) := T^{(2n-1)}(a_1,\ldots,a_{n-1})
\]
is a $2n$-gon map.
\end{Theorem}
\begin{proof}
This a special case of Theorem~\ref{thm:polygon->oddpolygon_red1}\,(2). See the proof there.
\end{proof}

\begin{Theorem}
For $n \in \bbN$, $n>1$, let $\tT^{(2n-1)}$ be a dual $(2n-1)$-gon map. Then
\[
 T^{(2n)}(a_1,\ldots,a_n) := \tT^{(2n-1)}(a_2,\ldots,a_n)
\]
is a $2n$-gon map.
\end{Theorem}
\begin{proof}
This a special case of Theorem~\ref{thm:polygon->dualpolygon_red2}\,(2).
\end{proof}

\begin{Theorem}
For $n \in \bbN$, $n>1$, let $\tT^{(2n)}$ be a dual $2n$-gon map. Then
\[
 T^{(2n+1)}(a_1,\ldots,a_n) := \tT^{(2n)}(a_2,\ldots,a_n)
\]
is a $(2n+1)$-gon map.
\end{Theorem}
\begin{proof}
This a special case of Theorem~\ref{thm:polygon->dualpolygon_red2}\,(2).
\end{proof}

\begin{Theorem}
For $n \in \bbN$, $n>1$, let $\tT^{(2n)}$ be a dual $2n$-gon map. Then
\[
 \tT^{(2n+1)}(a_1,\ldots,a_n) := \tT^{(2n)}(a_1,\ldots,a_{n-1})
\]
is a dual $(2n+1)$-gon map.
\end{Theorem}
\begin{proof}
This a special case of Theorem~\ref{thm:dualpolygon->evendualpolygon_red1}\,(2).
\end{proof}

\subsection{Further relations between neighboring polygon maps}
\label{subsec:further_rels}

\begin{Proposition}
Let $T_i^{(3)}$, $i=1,2$, be trigon maps. Then
\[
 \tT^{(4)}(a) = \big(T_1^{(3)}(a),T_2^{(3)}(a)\big)
\]
is a dual tetragon map if and only if the two trigon maps commute.
\end{Proposition}
\begin{proof}
This is easily verified. Also see Section~\ref{subsec:dual4gon}.
\end{proof}

\begin{Proposition}
For a map $T^{(5)}: \cU \times \cU \rightarrow \cU \times \cU$, let us write
\[
 T^{(5)}(a,b) =: \big(a \ast b, T^{(4)}(a,b)\big),
\]
with a binary operation $\ast$ and a map $T^{(4)} \colon \cU \times \cU \rightarrow \cU$. The following conditions are equivalent:
\begin{itemize}\itemsep=0pt
\item[$(1)$] $T^{(5)}$ is a pentagon map.
\item[$(2)$] $T^{(4)}$ is a tetragon map and, with $a \cdot b := T^{(4)}(a,b)$, the binary operations $\cdot$
and $\ast$ satisfy
\[
 (a \ast b) \ast ((a \cdot b) \ast c) = b \ast c, \qquad
 (a \ast b) \cdot ((a \cdot b) \ast c) = a \ast (b \cdot c),
\]
for all $a,b,c \in \cU$.
\end{itemize}
\end{Proposition}
\begin{proof}
\eqref{5gon_prod_conditions} shows that the conditions for $T^{(5)}$
to be a pentagon map are equivalent to $T^{(4)}$ being a tetragon map (which means associativity of $\cdot$)
and the two compatibility conditions for the two binary operations.
\end{proof}

\begin{Corollary}
Let $T^{(4)}$ be a tetragon map. Then
\[
 T^{(5)}(a,b) = \big(b, T^{(4)}(a,b)\big)
\]
is a pentagon map.
\end{Corollary}
\begin{proof}
Setting $a \ast b = b$ solves the two compatibility equations in condition (2) of the preceding proposition.
\end{proof}

\begin{Corollary}
Let $T^{(4)}$ be a tetragon map and $u$ a fixed element of $\cU$. Then
\[
 T^{(5)}(a,b) = \big(u, T^{(4)}(a,b)\big)
\]
is a pentagon map if and only if $T^{(4)}(u,u)=u$.
\end{Corollary}
\begin{proof}
Setting $a \ast b = u$ for all $a,b \in \cU$, solves the first of the two compatibility equations in condition (2) of the
preceding proposition and reduces the second to $u \cdot u =u$.
\end{proof}

\begin{Proposition}
Let $T^{(4)}$ be a tetragon map.
\begin{itemize}\itemsep=0pt
\item[$(1)$] The map
\[
 \tT^{(5)}(a,b) = \big(T^{(4)}(a,b),a\big)
\]
is a dual pentagon map.
\item[$(2)$] If $u$ is a fixed element of $\cU$, then
\[
 \tT^{(5)}(a,b) = \big(T^{(4)}(a,b),u\big)
\]
is a dual pentagon map if and only if $T^{(4)}(u,u)=u$.
\end{itemize}
\end{Proposition}
\begin{proof}
This is easily verified.
\end{proof}

\begin{Proposition}
For a map $\tT^{(6)} \colon \cU \times \cU \rightarrow \cU \times \cU \times \cU$, let us write
\[
 \tT^{(6)}(a,b) =: \big(T^{(5)}(a,b), a \diamond b\big),
\]
with a map $T^{(5)} \colon \cU \times \cU \rightarrow \cU \times \cU$ and a binary operation $\diamond$.
The following conditions are equivalent:
\begin{itemize}\itemsep=0pt
\item[$(1)$] $\tT^{(6)}$ is a dual hexagon map.
\item[$(2)$] $T^{(5)}$ is a pentagon map and, expressed as in \eqref{5gon_multiplications} $($so that \eqref{5gon_prod_conditions} holds$)$, it satisfies~\eqref{dualhex_conditions_beyond_pent}.
\end{itemize}
\end{Proposition}
\begin{proof}
This is easily verified.
\end{proof}

\begin{Example}
\label{ex:Rogers_dilog}
Let $\cdot$ be commutative and $a \diamond b := b \ast a$.
Then, as a consequence of \eqref{5gon_prod_conditions}, the additional dual hexagon map
conditions \eqref{dualhex_conditions_beyond_pent} are reduced to the single condition
\[
 ((a \cdot b) \ast c) \ast (a \ast b) = ((b \cdot c) \ast a) \ast (c \ast b).
\]
This holds, for example, if $\cU = (0,1) \subset \bbR$ and $a \ast b := (1-a) b/(1-ab)$.
Hence, using the pentagon map \eqref{5gon_sol_proj_geom}, the map
\[
 (a,b) \longmapsto \left( \frac{(1-a)b}{1-ab}, ab, \frac{a(1-b)}{1-ab} \right)
\]
solves the dual hexagon equation, also see \cite{Kash15}.
This map actually shows up in an identity for the
\emph{Rogers dilogarithm} function $L(a)$ \cite{Roge07}. $\cS(a) ={\rm e}^{\lambda L(a)}$,
with an arbitrary constant $\lambda \neq 0$, satisfies the pentagon relation\footnote{If we order
the parameters as $(a_0,\dots,a_4) := (a,1-ab,b,(1-b)/(1-ab),(1-a)/(1-ab))$, they are given by
the recursion relation $a_{n-1} a_{n+1} = 1-a_n$ (a special $Y$-system), which has
$\bbZ_5$ symmetry $a_{n+5} = a_n$ \cite{Fadd11,Glio+Tate95,Volk11}. }
\[
 \cS(b) \cS(a)
 = \cS\left( \frac{a(1-b)}{1-ab} \right) \cS(ab) \cS\left( \frac{(1-a)b}{1-ab} \right).
\]
Kashaev called a solution $\hat{T} \colon I \rightarrow \mathrm{End}(\cU \otimes \cU)$, where $I$ is the open unit interval
$(0,1) \subset \bbR$ and~$\cU$ a vector space, a matrix or operator dilogarithm if it satisfies the local
pentagon equation~\cite{Kash99TMP,Kash00TMP}
\[
 \hat{T}_{\bsy{23}}(a) \hat{T}_{\bsy{12}}(b)
 = \hat{T}_{\bsy{12}}\left( \frac{a(1-b)}{1-ab} \right) \hat{T}_{\bsy{13}}(ab) \hat{T}_{\bsy{23}}\left( \frac{(1-a)b}{1-ab} \right).
\]
\end{Example}

\begin{Corollary}
If $T^{(5)}$ is a pentagon map, then
\[
 \tT^{(6)}(a,b) := \big(T^{(5)}(a,b),a\big)
\]
is a dual hexagon map.
\end{Corollary}
\begin{proof}
Setting $a \diamond b := a$ solves the three equations \eqref{dualhex_conditions_beyond_pent}.
\end{proof}

\begin{Corollary}
\label{cor:pent->dhex,u}
Let $T^{(5)}$ be a pentagon map and $u \in \cU$ a fixed element. Then
\[
 \tT^{(6)}(a,b) := \big(T^{(5)}(a,b),u\big)
\]
is a dual hexagon map if and only if $T^{(5)}(u,u)=(u,u)$.
\end{Corollary}
\begin{proof}
Setting $a \diamond b := u$ reduces the three equations \eqref{dualhex_conditions_beyond_pent} to
$u \ast u = u$ and $u \cdot u = u$.
\end{proof}

\begin{Proposition}
\label{cor:dpent->dhex}
Let $\tT^{(5)}$ be a dual pentagon map.
\begin{itemize}\itemsep=0pt
\item[$(1)$] The map
\[
 \tT^{(6)}(a,b) = \big(b,\tT^{(5)}(a,b)\big)
\]
is a dual hexagon map.
\item[$(2)$] If $u$ is a fixed element of $\cU$, then
\[
 \tT^{(6)}(a,b) = \big(u,\tT^{(5)}(a,b)\big)
\]
is a dual hexagon map if and only if $\tT^{(5)}(u,u)=(u,u)$.
\end{itemize}
\end{Proposition}
\begin{proof}
This is also easily verified using results of Section~\ref{sec:dualPol}.
\end{proof}

\begin{Proposition}
For a map $T^{(7)}\colon \cU \times \cU \times \cU \rightarrow \cU \times \cU \times \cU$, let us write
\[
 T^{(7)}(a,b,c) =: (\{a,b,c\}, T^{(6)}(a,b,c)),
\]
with a ternary operation $\{ \,, \,, \}$ and a map $T^{(6)}\colon \cU \times \cU \times \cU \rightarrow \cU \times \cU$.
The following conditions are equivalent:
\begin{itemize}\itemsep=0pt
\item[$(1)$] $T^{(7)}$ is a heptagon map.
\item[$(2)$] $T^{(6)}$ is a hexagon map and, expressed as in \eqref{6gon_ternary_ops} $($so that \eqref{6gon_ternary_ops_cond} holds$)$, it satisfies the
compatibility conditions \eqref{7gon_conditions_beyond_6gon} with the above ternary operation.
\end{itemize}
\end{Proposition}
\begin{proof}
This is an immediate consequence of the last part of Section~\ref{subsec:7gon}.
\end{proof}

\begin{Corollary}
If $T^{(6)}$ is a hexagon map, then
\[
 T^{(7)}(a,b,c) := \big(c,T^{(6)}(a,b,c)\big)
\]
is a heptagon map.
\end{Corollary}
\begin{proof}
Using $\{a,b,c\}=c$ in \eqref{7gon_conditions_beyond_6gon} results in identities.
\end{proof}

\begin{Corollary}
Let $T^{(6)}$ be a hexagon map and $u \in \cU$ a fixed element. Then
\[
 T^{(7)}(a,b,c) := \big(u,T^{(6)}(a,b,c)\big)
\]
is a heptagon map if and only if $T^{(6)}(u,u,u) = (u,u)$.
\end{Corollary}
\begin{proof}
Using $\{a,b,c\}=u$ in \eqref{7gon_conditions_beyond_6gon} results in $\langle u,u,u \rangle = u = [u,u,u]$.
\end{proof}

\begin{Proposition}
Let $T^{(6)}$ be a hexagon map.
\begin{itemize}\itemsep=0pt
\item[$(1)$] The map
\[
 \tT^{(7)}(a,b,c) = \big(T^{(6)}(a,b,c),a\big)
\]
is a dual heptagon map.
\item[$(2)$] If $u$ is a fixed element of $\cU$, then
\[
 \tT^{(7)}(a,b,c) = \big(T^{(6)}(a,b,c),u\big)
\]
is a dual heptagon map if and only if $T^{(6)}(u,u,u)=(u,u)$.
\end{itemize}
\end{Proposition}
\begin{proof}
(1) Setting $[a,b,c] = a$, the last three equations of \eqref{dualhept_conditions} become identities and the
first three are equivalent to \eqref{6gon_ternary_ops_cond} by a renaming of the ternary operations.

(2) With $[a,b,c] = u$, the last three of equations \eqref{dualhept_conditions} become $T^{(6)}(u,u,u)=(u,u)$.
\end{proof}

\begin{Proposition}
For a map $\tT^{(8)}\colon \cU \times \cU \times \cU \rightarrow \cU \times \cU \times \cU \times \cU$, let us write
\[
 \tT^{(8)}(a,b,c) =: \big(T^{(7)}(a,b,c), |a,b,c| \big),
\]
with a map $T^{(7)}\colon \cU \times \cU \times \cU \rightarrow \cU \times \cU \times \cU$ and a ternary operation $|\,,\,, \, |$. The following conditions are equivalent:
\begin{itemize}\itemsep=0pt
\item[$(1)$] $\tT^{(8)}$ is a dual hexagon map.
\item[$(2)$] $T^{(7)}$ is a heptagon map and, expressed as in \eqref{7gon_ternary_ops}, it satisfies \eqref{dual8gon_beyond_hept}.
\end{itemize}
\end{Proposition}
\begin{proof}
This immediately follows from results in Section~\ref{subsec:dual_oct}.
\end{proof}

\begin{Corollary}
If $T^{(7)}$ is a heptagon map, then
\[
 \tT^{(8)}(a,b,c) := \bigl(T^{(7)}(a,b,c), a\bigr)
\]
is a dual octagon map.
\end{Corollary}
\begin{proof}
Setting $|a, b, c | = a$ turns the four equations \eqref{dual8gon_beyond_hept} into identities.
\end{proof}

\begin{Corollary}
Let $T^{(7)}$ be a heptagon map and $u \in \cU$ a fixed element. Then
\[
 \tT^{(8)}(a,b,c) := \big(T^{(7)}(a,b,c), u\big)
\]
is a dual octagon map if and only if $T^{(7)}(u,u,u) = (u,u,u)$.
\end{Corollary}
\begin{proof}
Setting $|a, b, c | = u$ turns the four equations \eqref{dual8gon_beyond_hept} into
$\{u,u,u\} = u$, $\langle u,u,u \rangle = u$ and~${[u,u,u] = u}$.
\end{proof}

\begin{Proposition}
Let $\tT^{(7)}$ be a dual heptagon map.
\begin{itemize}\itemsep=0pt
\item[$(1)$] The map
\[
 \tT^{(8)}(a,b,c) = \big(c,\tT^{(7)}(a,b,c)\big)
\]
is a dual octagon map.
\item[$(2)$] If $u$ is a fixed element of $\cU$, then
\[
 \tT^{(8)}(a,b,c) = \big(u,\tT^{(7)}(a,b,c)\big)
\]
is a dual octagon map if and only if $\tT^{(7)}(u,u,u)=(u,u,u)$.
\end{itemize}
\end{Proposition}
\begin{proof}
This can be verified using results of Section~\ref{sec:dualPol}.
\end{proof}

Preceding results suggest the following conjectures.

\begin{Conjecture}
Let $T^{(2n)}$ be a $2n$-gon map, $n \in \bbN$, $n \geq 2$. Then
\[
 T^{(2n+1)}(a_1,\ldots,a_n) := \big(a_n,T^{(2n)}(a_1,\ldots,a_n)\big)
\]
is a $(2n+1)$-gon map.
\end{Conjecture}

\begin{Conjecture}
Let $T^{(2n)}$ be a $2n$-gon map, $n \in \bbN$, and $u \in \cU$ a fixed element. Then
\[
 T^{(2n+1)}(a_1,\ldots,a_n) := \big(u, T^{(2n)}(a_1,\ldots,a_n)\big)
\]
is a $(2n+1)$-gon map if and only if $T^{(2n)}(u,\ldots,u) = (u,\ldots,u)$.
\end{Conjecture}

\begin{Conjecture}
Let $T^{(2n+1)}$ be a $(2n+1)$-gon map, $n \in \bbN$. Then
\[
 \tT^{(2n+2)}(a_1,\ldots,a_n) := \big(T^{(2n+1)}(a_1,\ldots,a_n),a_1\big)
\]
is a dual $(2n+2)$-gon map.
\end{Conjecture}

\begin{Conjecture}
Let $T^{(2n+1)}$ be a $(2n+1)$-gon map, $n \in \bbN$, and $u \in \cU$ a fixed element. Then
\[
 \tT^{(2n+2)}(a_1,\ldots,a_n) := \big(T^{(2n+1)}(a_1,\ldots,a_n),u\big)
\]
is a dual $(2n+2)$-gon map if and only if $T^{(2n+1)}(u,\ldots,u) = (u,\ldots,u)$.
\end{Conjecture}

\begin{Conjecture}
Let $\tT^{(2n+1)}$ be a dual $(2n+1)$-gon map, $n \in \bbN$. Then
\[
 \tT^{(2n+2)}(a_1,\ldots,a_n) := \big(a_n, \tT^{(2n+1)}(a_1,\ldots,a_n)\big)
\]
is a dual $(2n+2)$-gon map.
\end{Conjecture}

\begin{Conjecture}
Let $\tT^{(2n+1)}$ be a dual $(2n+1)$-gon map, $n \in \bbN$, and $u \in \cU$ a fixed element. Then
\[
 \tT^{(2n+2)}(a_1,\ldots,a_n) := \big(u, \tT^{(2n+1)}(a_1,\ldots,a_n)\big)
\]
is a dual $(2n+2)$-gon map if and only if $\tT^{(2n+1)}(u,\ldots,u) =(u,\ldots,u)$.
\end{Conjecture}

\section{Conclusions}
\label{sec:conclusions}
The main results of this work concern the structure of solutions of polygon equations. More precisely, we have shown
that a solution of a (dual) $N$-gon equation is related in very simple ways to solutions of the (dual) $(N+1)$-gon
and (dual) $(N-1)$-gon equation.
Each solution of a (dual) polygon equation extends to solutions of the higher equations.

For a chosen polygon equation, the most important case is when all basic sets are equal and there is only a
single polygon map. Expressing polygon equations with the help of transposition maps, we can quite easily
verify the above mentioned features for the simplest equations of the family. But for general proofs, we had to return
to the underlying framework of higher Tamari orders, as developed in \cite{DMH15}.

Our results reveal a beautiful structure of the family of polygon equations. Other nice aspects are
the integrability feature \cite{DMH15}, recalled in some examples in Section~\ref{sec:ExPEs}, and
relations with polyhedra \cite{DMH15}.

In this work, we concentrated on the set-theoretic setting. Many results directly pass over to the framework of
vector spaces, tensor products, and linear maps, where the dual tetragon equation becomes the coassociativity condition
and the pentagon equation plays one of its most important roles, as mentioned in the introduction.
A further exploration of polygon equations, beyond the pentagon equation, in this framework, will be left for
a separate work.

Also, concerning set-theoretic solutions, we have only set the stage. It is to be expected that more concrete solutions
can be obtained by applying methods that have already been exploited in the case of the pentagon equation.
Whereas the (dual) pentagon and dual hexagon equation can be expressed in terms of binary operations, the hexagon
and higher polygon equations involve~$n$-ary operations with $n>2$. There is a vast literature
 about such generalizations of products, and this should be helpful in finding solutions of $N$-gon equations
with $N>5$.
In particular, there are corresponding generalizations of co-, bi- and Hopf algebras, for which higher polygon
equations may play a role. Also see \cite{Stre98} (``cocycloids'').

The question pops up whether there are higher order counterparts of the important relation between solutions of
the pentagon equation and bi- or Hopf algebras, mentioned in the introduction. Since the transposition $\cP$ solves
the pentagon equation and also its dual, Corollaries \ref{cor:pent->dhex,u} and \ref{cor:dpent->dhex}
show that $\cP \otimes u$ and $u \otimes \cP$, where $u$ is a fixed element, are both dual hexagon maps.
If~$\cU$ is an algebra $\cA$, choosing
$u = 1_\cA$, these are algebra homomorphisms. They can be regarded as counterparts of the trivial
comultiplications $\Delta_\ell$ and $\Delta_r$ mentioned in the introduction. It should be obvious how this extends
to higher dual polygon equations, since a special solution of the odd $N$-gon equation is a combination of $\cP$'s
that achieves a total inversion. For example, since $\cP_{\mathrm{inv}} := \cP_{\bsy{1}} \cP_{\bsy{2}} \cP_{\bsy{1}}$ solves the
heptagon equation and its dual, it follows that $\cP_{\mathrm{inv}} \otimes 1_\cA$, as well as
$1_\cA \otimes \cP_{\mathrm{inv}}$, solve the dual octagon equation. A similarity transformation, analogous to
the construction of non-trivial comultiplications from a trivial one, as mentioned in the introduction, leads to a new algebra homomorphism. But under what conditions does it satisfy the respective (dual) even polygon equation?

As already mentioned in the introduction, solutions of simplex equations can be obtained from solutions of a polygon
equation and its dual \cite{DMH15,Kash+Serg98,Kass23,Mail94,Serg23}. But an additional compatibility
condition has to be solved, which needs further exploration. We plan to study this in a separate work.

We have seen in Section~\ref{subsec:further_rels} that a solution of the dual hexagon equation shows up in a~pentagonal relation for the (exponentiated) Rogers dilogarithm. Also higher (dual) polygon equations may play a
role in this context \cite{Volk97}. Surely there is much more to be revealed.

\subsection*{Acknowledgements}
Some important insights that led to this work originated from my collaboration with
Aristophanes Dimakis, who, sadly, passed away in~2021. I would like to thank Jim Stasheff for helpful correspondence.

\pdfbookmark[1]{References}{ref}
\LastPageEnding

\end{document}